\documentclass[twocolumn,journal]{IEEEtran}
\usepackage{amsmath,amsfonts}
\usepackage{algorithmic}
\usepackage{array}
\usepackage[caption=false,font=normalsize,labelfont=sf,textfont=sf]{subfig}
\usepackage{textcomp}
\usepackage{stfloats}
\usepackage{url}
\usepackage{verbatim}
\usepackage{graphicx}
\def\BibTeX{{\rm B\kern-.05em{\sc i\kern-.025em b}\kern-.08em
    T\kern-.1667em\lower.7ex\hbox{E}\kern-.125emX}}
\usepackage{balance}
\usepackage{hyperref}
\usepackage{xcolor}
\usepackage[numbers,sort&compress]{natbib}
\usepackage{amsthm}
\newtheorem{theorem}{Theorem} 

\usepackage[ruled]{algorithm2e}
\SetAlgorithmName{Algorithm}{Algorithm}{List of Algorithms}

\def\BibTeX{{\rm B\kern-.05em{\sc i\kern-.025em b}\kern-.08em
		T\kern-.1667em\lower.7ex\hbox{E}\kern-.125emX}}
\newcommand{\argmax}{\operatorname{arg\,max}}
\newcommand{\argmin}{\operatorname{arg\,min}}

\begin{document}
\title{Low Complexity Detector for XL-MIMO Uplink: A Cross Splitting Based Information Geometry Approach}
\author{\IEEEauthorblockN{Wenjun Zhang, An-An Lu\textsuperscript{*}, \textit{Member, IEEE} and Xiqi Gao, \textit{Fellow, IEEE}}
\thanks{Wenjun Zhang, An-An Lu, and Xiqi Gao are with the National Mobile Communications Research Laboratory, Southeast University, Nanjing
	210096 China, and also with Purple Mountain Laboratories, Nanjing 211111, China (e-mail: 220231059@seu.edu.cn; aalu@seu.edu.cn;
	xqgao@seu.edu.cn).

}}


\maketitle

\begin{abstract}
In this paper, we propose the cross splitting based information geometry approach (CS-IGA), a novel and low complexity iterative detector for uplink signal recovery in extra-large-scale MIMO (XL-MIMO) systems. Conventional iterative detectors, such as  the approximate message passing (AMP) algorithm and the traditional information geometry algorithm (IGA), suffer from a per iteration complexity that scales with the number of base station (BS) antennas, creating a computational bottleneck. To overcome this, CS-IGA introduces a novel cross matrix splitting of the natural parameter in the \emph{a posteriori} distribution. This factorization allows the iterative detection based on the  matched filter, which reduces per iteration computational complexity. Furthermore, we extend this framework to nonlinear detection and propose nonlinear CS-IGA (NCS-IGA) by seamlessly embedding discrete constellation constraints, enabling symbol-wise processing without external interference cancellation loops. Comprehensive simulations under realistic channel conditions demonstrate that CS-IGA matches or surpasses the bit error rate (BER) performance of Bayes optimal AMP and IGA for both linear and nonlinear detection, while achieving this with fewer iterations and a substantially lower computational cost. These results establish CS-IGA as a practical and powerful solution for high-throughput signal detection in next generation XL-MIMO systems.
\end{abstract}


\section{Introduction}
\IEEEPARstart{M}{assive} multiple-input multiple-output (MIMO) has emerged as a cornerstone for achieving the spectral and energy efficiency targets of the fifth-generation (5G) cellular networks \cite{Marzetta2010,mimo3,mimo1,mimo2}. 
Building on this success, its evolved version extra-large-scale MIMO (XL-MIMO) envisions antenna apertures extending beyond tens of meters to provide unprecedented spatial resolution, capacity, and reliability in the sixth-generation (6G) wireless communications systems \cite{xlmimo,xlmimo2,xlmimo3}. In this context, accurate and efficient signal detection becomes increasingly critical. This paper addresses the signal detection challenge in XL‑MIMO deployments.

Concurrently, 3GPP has identified frequencies above 6 GHz (the U6G band) as a primary allocation for future wireless access \cite{3GPP38300}. In this band, the enormous spatial degrees of freedom provided by extra‑large antenna apertures facilitate precise user separation and interference suppression. While XL-MIMO will employ an extra-large array with hundreds or thousands of antennas and transmit up to hundreds of data streams \cite{hundred_streams}, the vast scale of these arrays and streams dramatically increases the computational demands of uplink multi-user detection. While optimal detectors such as maximum \textit{a posteriori} (MAP) and maximum likelihood (ML) offer the lowest error rates, they entail an exhaustive search over the joint symbol space whose complexity grows exponentially with the number of users \cite{MLMAP}. As a result, MAP and ML detection become impractical in XL‑MIMO systems.


Linear receivers, particularly the zero forcing (ZF) and the linear minimum mean square error (LMMSE), offer a favorable trade-off between performance and complexity by reducing detection to matrix-vector multiplications plus a single matrix inversion \cite{LinearDet}. However, as the number of antennas and users increases, even one inversion can become computationally intensive. To alleviate this burden, a variety of low complexity inversion approximations have been developed. For instance, the Neumann series expansion was applied in \cite{NS} to approximate the inversion, while Newton’s iterative method offering faster convergence was explored in \cite{NI}. The Jacobi iteration, introduced in \cite{Jaccobi}, further reduces per iteration cost, and a Gauss-Seidel based detector in \cite{Gauss-seidel} achieves a similar complexity with an improved convergence rate. Additionally, the conjugate gradient (CG) algorithm has been shown to approach near optimal performance when the base-station-to-user-antenna ratio is high \cite{CG}.

For nonlinear detection, Bayesian inference methods, such as belief propagation (BP) \cite{BP} and expectation propagation (EP) \cite{EP} leverage prior distributions on both channels and symbols to more effectively suppress multi-user interference, often yielding superior performance over linear receivers when high order constellations are used. Moreover, variational Bayes frameworks have been employed for joint detection and decoding \cite{Beal2003}, trading off enhanced error rate performance for additional iterative complexity.

More recently, the approximate message passing (AMP) algorithm has gained attention as an attractive detection framework, offering low per iteration complexity and near optimal performance under certain random matrix assumptions \cite{AMP}. AMP variants, such as the orthogonal AMP (OAMP) \cite{OAMP}, the vector AMP (VAMP) \cite{VAMP}, and the memory AMP (MAMP) \cite{MAMP}, extend the original algorithm to cope with ill conditioned channel matrices or discrete symbol priors, albeit with a modest increase in computational load. In many practical systems, user scheduling precedes detection to ensure well conditioned channel gains. Under these circumstances, standard AMP strikes an excellent balance between complexity and performance, making it particularly well suited for XL‑MIMO signal detection.

Information geometry (IG) offers an alternative framework that has been recently applied in wireless communications. For example, the decoding algorithms for turbo and low density parity check (LDPC) codes can be derived from an IG perspective, with analyses of their equilibrium points and error behavior in \cite{IGForLDPC}. The information geometry approach (IGA) for channel estimation and symbol detection, developed in \cite{IGChEst} and \cite{IGSiDet}, matches the performance of AMP while providing a clear geometric interpretation for each update. Simplified IG variants, such as the simplified IGA (SIGA) for channel estimation \cite{SIGA}, and the approximate IGA (AIGA) demonstrating connections between IGA and AMP \cite{IGForAMP}, further unify message-passing and geometric viewpoints. More recently, an IGA formulation in \cite{ICIGA} reduces per-iteration complexity so that it scales with the number of users, rather than the number of BS antennas. However, its simplified form is the same as the classical Jacobi method, which converges typically slower than AMP in signal detection.

Although IG-based detectors and Bayes optimal AMP variants have great performance in massive MIMO, they incur high per-iteration complexities in XL-MIMO because both the antenna array and data streams become extremely large.
Furthermore, the shift toward very high carrier frequencies and ultra‑dense user deployments in 6G make algorithms that scale more gracefully with system dimensions more necessary. 
For near-field low-complexity detectors in XL-MIMO, the visibility region (VR) is utilized to develop a randomized Kaczmarz algorithm \cite{VR1} and partial zero-forcing \cite{VR2} for XL-MIMO multi-user detection. Nevertheless, these methods may lack applicability in far-field communication. As the visibility region expands, it inherently induces a rise in complexity, thereby posing challenges to the implementation of the methods.

Motivated by the dual goals of rapid convergence and minimizing whole array detection complexity in far field, we propose a novel information geometry based detector that employs a cross splitting of a matrix in the \emph{a posteriori} probability distributions.  Our main contributions are:
\begin{itemize}
	\item  We propose a novel cross splitting framework to derive the information geometry method for detection, so that the complexity per iteration grows only linearly with the number of users.
	
	\item  We derive the cross splitting based information geometry approach (CS-IGA) whose convergence rate rivals that of Bayes optimal AMP, yet whose computational cost per iteration is substantially lower.  We further prove that the algorithm  converges to the MMSE detector.
	
	\item For discrete symbol priors, we propose a nonlinear CS‑IGA variant (NCS‑IGA) that reproduces the convergence behavior of Bayes optimal AMP while reducing per iteration cost.
	
\end{itemize}

The remainder of this paper is structured as the following parts. Section \ref{Information Geometry} formulates the XL-MIMO detection problem and reviews essential aspects of information geometry. In Section \ref{CS-IGA Linear} presents the cross splitting construction, derives the CS‑IGA algorithm, and analyzes its convergence and fixed point properties. Section \ref{CS_IGA Non-Linear} extends CS-IGA to the nonlinear detection setting. Simulation results are shown in Section \ref{sim}, and Section \ref{conclusion} concludes the paper.

\textbf{Notations}: We use boldface upper‑case letters (e.g., $\mathbf{A}$, $\mathbf{H}$) for matrices and boldface lower‑case letters (e.g., $\mathbf{x}$, $\mathbf{y}$) for column vectors. The superscript $(\cdot)^*$ denotes complex conjugation and $(\cdot)^H$ denotes Hermitian transpose. For a matrix $\mathbf{H}$, its $(m,n)$ entry is $H_{mn}$. The identity matrix of size $N$ is $\mathbf{I}_N$. For a vector $\pmb{a} = [a_1, a_2, \dots, a_N]^T$, we write $a_n$ for its $n$-th element and $\bar{\pmb{a}}_n$ for the $(N-1)$-dimensional vector obtained by removing $a_n$. The expectation and variance operators are $\mathbb{E}[\cdot]$ and $\mathbb{V}[\cdot]$, respectively. The symbol $\odot$ denotes element‑wise multiplication. Define $\mathcal{Z}_n\equiv\left\{0,1,\dots,N\right\}$ as the natural number set and $\mathcal{Z}_n^+\equiv\left\{1,2,\dots,N\right\}$ as the positive integer set. Denote the complex number field as $\mathbb{C}$ and the real number field as $\mathbb{R}$. $\mathbf{P}_{1n}$ denotes the permutation matrix that brings the $n$-th row to the first position and keep the row-order invariant.

\section{Problem Formulation and Information Geometry}\label{Information Geometry}

\subsection{System Model}
Consider the uplink of an XL-MIMO system in which a base station (BS) is equipped with \(M\) antennas and serves \(N\) single-antenna users. Let
\[
\mathbf{x} = [\,x_1,\,x_2,\,\dots,\,x_N]^T \in \mathbb{C}^{N}
\]
be the vector of coded and modulated transmit symbols, where each \(x_n\) is drawn from a unit-energy constellation, e.g., \(L\)-QAM, and satisfies \(\mathbb{E}[|x_n|^2]=1\) for \(n\in\mathcal{Z}_N^+\).

The received signal at the BS is modeled as
\begin{equation}
	\label{eq:system_model}
	\mathbf{y} = \mathbf{H}\mathbf{x} + \mathbf{z},
\end{equation}
where \(\mathbf{y}\in\mathbb{C}^{M}\) denotes the received signal vector, \(\mathbf{H}\in\mathbb{C}^{M\times N}\) represents the uplink channel matrix, and \(\mathbf{z}\sim\mathcal{CN}(\mathbf{0},\sigma_z^2\mathbf{I}_M)\) is a circularly symmetric complex Gaussian noise. In this paper, we assume that perfect channel state information (CSI) is available at the BS.

\subsection{Problem Statement}

\subsubsection{Linear Detection Scheme}
We assume that the transmitted symbols \(\{x_n\}\) are mutually independent and uncorrelated with the noise vector \(\mathbf{z}\). By adopting a Gaussian prior for \(\mathbf{x}\), the \emph{a posteriori} density becomes
\begin{equation}
	\label{eq:posterior}
	p(\mathbf{x}\mid\mathbf{y})
	 \propto 
	\exp \Bigl(-\tfrac{1}{\sigma_z^2}\|\mathbf{y}-\mathbf{H}\mathbf{x}\|^2  - \mathbf{x}^H\mathbf{R}_{x}\mathbf{x}\Bigr),
\end{equation}
where \(\mathbf{R}_{x}=\mathbb{E}[\mathbf{x}\mathbf{x}^H]\). Under the standard independent–unit‐variance assumption, we set \(\mathbf{R}_{x}=\mathbf{I}_N\). Then, the \emph{a posteriori} mean and variance are denoted as \cite{Kay1993}
\begin{subequations}
	\begin{align}
		\label{eq:posterior mean}
		\boldsymbol{\mu}_{\mathrm{LMMSE}}
		&= \bigl(\sigma_z^{-2}\mathbf{H}^H\mathbf{H} + \mathbf{I}_N\bigr)^{-1}
		\bigl(\sigma_z^{-2}\mathbf{H}^H\mathbf{y}\bigr),\\
		\label{eq:posterior variance}
		\boldsymbol{\Sigma}_{\mathrm{LMMSE}}
		&= \bigl(\sigma_z^{-2}\mathbf{H}^H\mathbf{H} + \mathbf{I}_N\bigr)^{-1}.
	\end{align}
\end{subequations}
They are equivalent to the LMMSE detection results. These estimates are supplied to a soft demodulator, whose outputs are subsequently decoded. 

Direct inversion of the \(N\times N\) matrix \((\sigma_z^{-2}\mathbf{H}^H\mathbf{H} + \mathbf{I}_N)\) requires  computational complexity of order \(\mathcal{O}(N^3)\),  which is prohibitive for practical XL‑MIMO systems where the number of users might up to more than $100$.

\subsubsection{Non‑Linear Detection Scheme}
Starting from the received signal in \eqref{eq:system_model}, when \(\mathbf{x}\) follows a discrete prior rather than Gaussian, the \emph{a posteriori} distribution is written as
\begin{IEEEeqnarray}{Cl}
	\label{eq:posterior non-liner}
	p(\mathbf{x}\mid\mathbf{y})
	&\propto
	p_{pr}(\mathbf{x})p(\mathbf{y}\mid\mathbf{x}) \notag\\
	&\propto\prod_{n=1}^{N}p_{pr,n}({x}_n)\exp\Bigl(-\frac{1}{\sigma_z^2}\|\mathbf{y}-\mathbf{H}\mathbf{x}\|^2 \Bigr),
\end{IEEEeqnarray}
where the prior factorizes as \(p_{pr}(\mathbf{x})=\prod_{n=1}^{N}p_{pr,n}(x_n)\). For a uniform \(L\)-point constellation \(\{x^{(\ell)}\}_{\ell=1}^L\), each prior probability is
\begin{IEEEeqnarray}{Cl}
	p_{pr,n}(x_n)\bigl|_{x_n = x^{(\ell)}} = \frac{1}{L}, 
	\quad
	\{x^{(\ell)}\}_{\ell=1}^L\in\mathbb{X}^{L}.
\end{IEEEeqnarray}
The MAP detector then selects
\begin{IEEEeqnarray}{Cl}
	\label{eq:MAP}
	\hat{\mathbf{x}}_{MAP} = \underset{{\mathbf{x}\in\mathbb{X}^L}}{\argmax}~p(\mathbf{x}\mid\mathbf{y}),
\end{IEEEeqnarray}
which minimizes the symbol error rate but entails an exhaustive search over \(\mathbb{X}^L\). Since this complexity grows exponentially with \(N\), exact MAP detection is infeasible for XL‑MIMO deployments.

\subsection{Definitions from Information Geometry}
In this subsection, we introduce the definitions of statistical manifold of complex Gaussians, the objective manifold, the $m$-projection and the auxiliary manifolds    that will be used in the derivation of the new IGA algorithm.


\subsubsection{Statistical Manifold of Complex Gaussians}
Let $\mathcal{M}=\bigl\{p(\mathbf{x};\boldsymbol{\theta},\boldsymbol{\Theta})\bigr\}$ denote the family of circularly symmetric complex Gaussian distributions over $\mathbf{x}\in\mathbb{C}^N$ with a full-rank $\boldsymbol{\Theta}$, defined as
\begin{IEEEeqnarray}{rCl}
	p(\mathbf{x};\boldsymbol{\theta},\boldsymbol{\Theta})
	&=& \exp\Bigl\{{\mathbf{x}^H\boldsymbol{\theta} + \boldsymbol{\theta}^H\mathbf{x}}
	\;+\;{\mathbf{x}^H\boldsymbol{\Theta}\,\mathbf{x}}
	\;-\;\psi\Bigr\},
	\label{eq:natparam}
\end{IEEEeqnarray}
where
\begin{IEEEeqnarray}{Cl}
	\psi(\boldsymbol{\theta},\boldsymbol{\Theta})
	&= \log\!\int_{\mathbb{C}^N}
	\exp\bigl\{\mathbf{x}^H\boldsymbol{\theta} + \boldsymbol{\theta}^H\mathbf{x}
	+ \mathbf{x}^H\boldsymbol{\Theta}\,\mathbf{x}\bigr\} \,d\mathbf{x} \notag\\
	&= N\log(\pi)-\log\det(-\boldsymbol{\Theta})-\boldsymbol{\theta}^H\boldsymbol{\Theta}^{-1}\boldsymbol{\theta},
\end{IEEEeqnarray}
is the normalization term and is also called free-energy.  Here $\boldsymbol{\Theta}\prec\mathbf{0}$ is negative‐definite and $\boldsymbol{\theta}\in\mathbb{C}^N$.  The pair $(\boldsymbol{\theta},\boldsymbol{\Theta})$ defines the \emph{e‐}coordinates of the point on $\mathcal{M}$ and are called natural parameters.

Correspondingly, one may adopt the \emph{m‐}coordinates in the affine coordinate system. Given the \emph{e-}coordinates, the \emph{m-}coordinates are calculated as
\begin{IEEEeqnarray}{rClCl}
	\boldsymbol{\mu} = \mathbb{E}[\mathbf{x}]
	&=& -\,\boldsymbol{\Theta}^{-1}\,\boldsymbol{\theta},
	\label{eq:mcoord_mean}\\
	\boldsymbol{\Sigma} + \boldsymbol{\mu}\boldsymbol{\mu}^H
	&=& \mathbb{E}[\mathbf{x}\mathbf{x}^H].
	\label{eq:mcoord_cov}
\end{IEEEeqnarray}
With abuse of notation, we also call $(\boldsymbol{\mu}, \boldsymbol{\Sigma})$ the the \emph{m-}coordinates.
The \emph{m-}coordinates are also called as the expectation parameters in information geometry. The dual function of $\phi(\boldsymbol{\theta},\boldsymbol{\Theta})$ is defined as
\begin{IEEEeqnarray}{Cl}
	\phi(\boldsymbol{\theta},\boldsymbol{\Theta}) = \int p(\mathbf{x};\boldsymbol{\theta},\boldsymbol{\Theta})\log p(\mathbf{x};\boldsymbol{\theta},\boldsymbol{\Theta})~d\mathbf{x},
\end{IEEEeqnarray}
which is also called negative entropy. By using the dual affine coordinate system, $\phi$ is rewritten as
\begin{IEEEeqnarray}{Cl}
	\phi(\boldsymbol{\mu},\boldsymbol{\Sigma})=-\log\det{\boldsymbol{\Sigma}}+c,
\end{IEEEeqnarray}
where $c$ is a constant value. It is noticeable that $\phi(\boldsymbol{\mu},\boldsymbol{\Sigma})$ is a convex function of $\boldsymbol{\mu}$ and $\boldsymbol{\Sigma}$, and induces the affine structure.

The duality between \emph{e‐} and \emph{m‐}coordinates is governed by the Legendre transform  and obtained from the derivatives of $\phi$ and $\psi$ as
\begin{IEEEeqnarray}{rCl}
	\boldsymbol{\theta} &=& \boldsymbol{\Sigma}^{-1}\,\boldsymbol{\mu},
	\quad
	\boldsymbol{\Theta} = -\,\boldsymbol{\Sigma}^{-1},
	\label{eq:legendre1}\\
	\boldsymbol{\mu} &=& -\,\boldsymbol{\Theta}^{-1}\,\boldsymbol{\theta},
	\quad
	\boldsymbol{\Sigma} = -\,\boldsymbol{\Theta}^{-1}.
	\label{eq:legendre2}
\end{IEEEeqnarray}

The \textit{a posteriori} density of $\mathbf{x}$ for linear detection in \eqref{eq:posterior} can be represented as 
\begin{IEEEeqnarray}{rCl}
	\label{eq:linear posterior}
	p(\mathbf{x}\mid\mathbf{y})
	&\propto& \exp \Bigl\{\mathbf{x}^H\underbrace{\sigma_z^{-2}\mathbf{H}^H\mathbf{y}}_{\boldsymbol{\theta}_{\mathrm{post}}}
	\;+\;\underbrace{\mathbf{y}^H\mathbf{H}\,\sigma_z^{-2}}_{\boldsymbol{\theta}_{\mathrm{post}}^H}\mathbf{x}\notag\\
	&\;+\;&\mathbf{x}^H\underbrace{[-(\sigma_z^{-2}\mathbf{H}^H\mathbf{H}+\mathbf{I})]}_{\boldsymbol{\Theta}_{\mathrm{post}}}\mathbf{x}\Bigr\}.
	\label{eq:posterior_mfld}
\end{IEEEeqnarray}
Hence, it is a point in $\mathcal{M}$ and its \emph{e‐}coordinates are
\begin{IEEEeqnarray}{rCl}
	\boldsymbol{\theta}_{\mathrm{post}}
	&=& \sigma_z^{-2}\,\mathbf{H}^H\mathbf{y},
	\quad
	\boldsymbol{\Theta}_{\mathrm{post}}
	= -\,\bigl(\sigma_z^{-2}\mathbf{H}^H\mathbf{H} + \mathbf{I}\bigr).
	\label{eq:post_nat}
\end{IEEEeqnarray}
By \eqref{eq:legendre2}, the exact \emph{a posteriori} mean $\boldsymbol{\mu}_{\mathrm{post}}$ and covariance $\boldsymbol{\Sigma}_{\mathrm{post}}$ follow as
\begin{IEEEeqnarray}{rCl}
	\boldsymbol{\mu}_{\mathrm{post}}
	&=& -\,\boldsymbol{\Theta}_{\mathrm{post}}^{-1}\,\boldsymbol{\theta}_{\mathrm{post}},
	\quad
	\boldsymbol{\Sigma}_{\mathrm{post}}
	= -\,\boldsymbol{\Theta}_{\mathrm{post}}^{-1} 
	\label{eq:post_mcoord}
\end{IEEEeqnarray}
which is the same as in \eqref{eq:posterior mean} and \eqref{eq:posterior variance}. 

\subsubsection{Objective Manifold of Complex Gaussians}
To mitigate the inversion complexity, we define the {objective manifold} as
\begin{IEEEeqnarray}{rCl}
	\mathcal{M}_0 
	&=& \bigl\{\,p(\mathbf{x};\boldsymbol{\theta}_0,\boldsymbol{\Theta}_0)\bigr\},
	\label{eq:obm_def}
\end{IEEEeqnarray}
where $\boldsymbol{\Theta}_0$ is diagonal so that inversion complexity reduces to $O(N)$.  Its \emph{m‐}coordinates are
\begin{IEEEeqnarray}{rCl}
	\boldsymbol{\mu}_0 = -\,\boldsymbol{\Theta}_0^{-1}\,\boldsymbol{\theta}_0,
	\quad
	\boldsymbol{\Sigma}_0 = -\,\boldsymbol{\Theta}_0^{-1}.
	\label{eq:obm_mcoord}
\end{IEEEeqnarray}

Let $(\boldsymbol{\theta}_{\mathrm{post}},\boldsymbol{\Theta}_{\mathrm{post}})$ be the true \emph{a posteriori} $e$-coordinates in the linear detection case.  Our goal is to find a proper point with $e$-coordinate $(\boldsymbol{\theta}_0,\boldsymbol{\Theta}_0)$ in $\mathcal{M}_0$ to ensure its $m$-coordinate can approximate the \emph{a posteriori} mean and variance, \textit{i.e.},
\begin{IEEEeqnarray}{Cl}
	\label{eq:approximation_OBM}
	\boldsymbol{\mu}_0 \approx \boldsymbol{\mu}_{\mathrm{post}},
	\quad
	\mathbf{I}\odot \boldsymbol{\Sigma}_0  \approx  \mathbf{I}\odot \boldsymbol{\Sigma}_{\mathrm{post}} 
\end{IEEEeqnarray}
or even 
\begin{IEEEeqnarray}{Cl}
	\label{eq:approximation_OBM}
	\boldsymbol{\mu}_0 = \boldsymbol{\mu}_{\mathrm{post}},
	\quad
	\mathbf{I}\odot \boldsymbol{\Sigma}_0  = \mathbf{I}\odot \boldsymbol{\Sigma}_{\mathrm{post}}.
\end{IEEEeqnarray}
For nonlinear detection case, the definition of the objective manifold remains the same, only the definitions of  $\boldsymbol{\mu}_{\mathrm{post}}$ and $\boldsymbol{\Sigma}_{\mathrm{post}}$ change.

\subsubsection{Kullback–Leibler (KL) Divergence and {m-}Projection}
Denote two points in a manifold $\mathcal{M}$ as $P:p(\mathbf{x};\boldsymbol{\theta}_1,\boldsymbol{\Theta}_1)$  and $Q:p(\mathbf{x};\boldsymbol{\theta}_2,\boldsymbol{\Theta}_2)$ . The KL divergence from $P$ to $Q$ is defined as
\begin{IEEEeqnarray}{Cl}
	\mathcal{D}_{\mathrm{KL}}(Q;P) = \mathcal{D}_{\mathrm{KL}}(p(\mathbf{x};\boldsymbol{\theta}_2,\boldsymbol{\Theta}_2);p(\mathbf{x};\boldsymbol{\theta}_1,\boldsymbol{\Theta}_1)),
\end{IEEEeqnarray}
which can be rewritten in the dual affine coordinate systems as \cite{amari}
\begin{IEEEeqnarray}{Cl}
	\mathcal{D}_{\mathrm{KL}}(Q;P)&= \phi(\boldsymbol{\mu}_2,\boldsymbol{\Sigma}_2)+\psi(\boldsymbol{\theta}_1,\boldsymbol{\Theta}_1)-\boldsymbol{\mu}_2^H\boldsymbol{\theta}_1 \notag \\
	&-~\boldsymbol{\theta}_1^H\boldsymbol{\mu}_2-\mathrm{tr}((\boldsymbol{\Sigma}_2+\boldsymbol{\mu}_2\boldsymbol{\mu}_2^H)\boldsymbol{\Theta}_1).
\end{IEEEeqnarray}

Before introducing the \emph{m}-projection, we first define \emph{e-}flat and \emph{m-}flat submanifolds. A submanifold $\mathcal{M}_s \subset \mathcal{M}$ is \emph{e-}flat if it is defined by a set of linear constraints in the \emph{e-}coordinates $(\boldsymbol{\theta}, \boldsymbol{\Theta})$. Likewise, $\mathcal{M}_s$ is \emph{m-}flat if it is defined by linear constraints in the \emph{m-}coordinates $(\boldsymbol{\mu}, \boldsymbol{\Sigma})$.

The projection of a distribution onto an \emph{e-}flat submanifold is called an \emph{m-}projection, since it is realized by a linear operation in the dual affine (\emph{m-}) coordinate system. Let $P \in \mathcal{M}_s$ and $Q \notin \mathcal{M}_s$. The \emph{m-}projection of $Q$ onto $\mathcal{M}_s$ is given by
\begin{IEEEeqnarray}{Cl}
	p(\mathbf{x};\boldsymbol{\theta}_2^0,\boldsymbol{\Theta}_2^0) &= \Pi_{\mathcal{M}_s}^m\left\{p(\mathbf{x};\boldsymbol{\theta}_2,\boldsymbol{\Theta}_2)\right\} \notag\\
	&= \underset{p(\mathbf{x};\boldsymbol{\theta}_1,\boldsymbol{\Theta}_1)\in\mathcal{M}_s}{\argmin}\mathcal{D}_{\mathrm{KL}}(Q;P).
\end{IEEEeqnarray}

Now consider the case where $\mathcal{M}_s$ is the objective manifold, i.e., $\mathcal{M}_s = \mathcal{M}_0$, in which case the precision matrix $\boldsymbol{\Theta}_1$ is constrained to be diagonal. The partial derivative of the KL divergence with respect to the natrual parameters of the \emph{m-}projection point is then given by
\begin{IEEEeqnarray}{Cl}
	\frac{\partial \mathcal{D}_{KL}(Q;P)}{\partial \boldsymbol{\theta}_1^*}&=\frac{\partial \psi(\boldsymbol{\theta}_1,\boldsymbol{\Theta}_1)}{\partial \boldsymbol{\theta}_1^*}-\boldsymbol{\mu}_2 \notag\\
	&=-\boldsymbol{\Theta}_1^{-1}\boldsymbol{\theta}_1-\boldsymbol{\mu}_2=\boldsymbol{\mu}_1-\boldsymbol{\mu}_2,\IEEEyesnumber\IEEEyessubnumber*\\
	\frac{\partial \mathcal{D}_{KL}(Q;P)}{\partial \boldsymbol{\Theta}_1}&=\frac{\partial \psi(\boldsymbol{\theta}_1,\boldsymbol{\Theta}_1)}{\partial \boldsymbol{\Theta}_1}-\frac{\partial \mathrm{tr}(\mathbf{V}_2\boldsymbol{\Theta}_1)}{\partial\boldsymbol{\Theta}_1} \notag\\
	&=\mathbf{I}\odot(\boldsymbol{\mu}_1\boldsymbol{\mu}_1^H-\boldsymbol{\Theta}_1^{-1})- \mathbf{I}\odot\mathbf{V}_2 \notag\\
	&=\mathbf{I}\odot\mathbf{V}_1-\mathbf{I}\odot\mathbf{V}_2,
\end{IEEEeqnarray}
where $\mathbf{V}_1=\boldsymbol{\Sigma}_1+\boldsymbol{\mu}_1\boldsymbol{\mu}_1^H$ and $\mathbf{V}_2=\boldsymbol{\Sigma}_2+\boldsymbol{\mu}_2\boldsymbol{\mu}_2^H$. Thus, we have the following property of the \emph{m-}projection
\begin{IEEEeqnarray}{Cl}
\label{eq:m_projection}
	\boldsymbol{\mu}_1 &= \boldsymbol{\mu}_2,\IEEEyesnumber\IEEEyessubnumber*\\
	\mathbf{I}\odot\boldsymbol{\Sigma}_1&=\mathbf{I}\odot\boldsymbol{\Sigma}_2.
\end{IEEEeqnarray}

\subsubsection{Auxiliary Manifolds}

In the linear detection case, computing the \emph{a posteriori} mean and variances of the transmitted symbols is the same as performing the $m$-projection. However, a direct computing requires the full inverse of $\boldsymbol{\Theta}_{\mathrm{post}}$.
In the nonlinear detection case, even more complicated calculations are needed to obtain a similar $m$-projection.

To avoid computing the $m$-projection of the \emph{a posteriori} distribution directly, a sequence of carefully constructed {auxiliary manifolds} is introduced to achieve the approximate or exact $m$-projection with much lower complexity. The role of constructing these auxiliary manifold is the same as the factorization of a probability density function (PDF) in the message passing algorithm, but in a different view.  The detailed information geometry framework is provided in \cite{ICIGA}.

Since different constructing of auxiliary manifolds will lead to different algorithms, the specific constructing method is the key in deriving the IGA.
In traditional IGA frameworks \cite{IGChEst}, the number of auxiliary manifolds equals the number of BS antennas, which entails a per-iteration computational complexity proportional to the number of BS antennas and is prohibitive in practical XL-MIMO systems. 

In the next section, we propose a novel auxiliary manifolds construction that achieves fast convergence in the XL-MIMO detection while significantly reducing the number of required auxiliary manifolds. Based on this construction, we derive the CS-IGA.  It might also be possible to apply the proposed construction to the message passing framework to derive a new algorithm equivalent to the proposed algorithm. However, we prefer the IG framework since the new construction looks more natural under it.

\section{CS-IGA for Linear Detection}\label{CS-IGA Linear}
In this section, we present the CS-IGA for approximate LMMSE detection in XL-MIMO systems. We begin by a new construction of the auxiliary manifolds. Next, we provide a detailed derivation of the CS-IGA. We then prove that the proposed algorithm converges to the LMMSE estimation. Finally, we analyze the computational complexity of the proposed approach.

\subsection{Construction of Auxiliary Manifolds with Cross Splitting}

Starting from the received signal model in \eqref{eq:system_model} and under the Gaussian assumption for \(\mathbf{x}\) in the linear case, the exact \emph{a posteriori} can be expressed in its natural‐parameter form as in \eqref{eq:post_nat}.  We begin by decomposing the precision matrix
\begin{IEEEeqnarray}{Cl}
	\mathbf{K} = \frac{1}{\sigma_z^2}\mathbf{H}^H\mathbf{H} +  \mathbf{I} 
	= \bar{\mathbf{K}}  + \mathbf{I}\odot\mathbf{K},\IEEEyesnumber\IEEEyessubnumber*\\
	\quad
	\bar{K}_{ij} = 
	\begin{cases}
		K_{ij}, & i\neq j,\\
		0,      & i=j.
	\end{cases}
\end{IEEEeqnarray}
Let \(\mathbf{D} =  \mathbf{I}\odot\mathbf{K}\) collect the diagonal entries.  For each \(n=1,\dots,N\), define the off‐diagonal vector as
\[
\bar{\mathbf{k}}_n 
= \tfrac12\bigl[\,K_{1n},\dots,K_{(n-1)n},K_{n(n+1)}^*,\dots,K_{nN}^*\,\bigr]^H
\]
and let \(\mathbf{P}_{1n}\) be the permutation matrix that moves index \(n\) to the first position.

We then perform a splitting of the  natural parameter $\boldsymbol{\Theta}_{\mathrm{post}}$ of the \emph{a posteriori} PDF into \(N\) low‐rank components plus a diagonal term as
\begin{IEEEeqnarray}{Cl}
	\label{eq:splitting}
	\boldsymbol{\theta}_{\mathrm{post}}
	&= \sum_{n=1}^N \mathbf{b}_n,\IEEEyesnumber\IEEEyessubnumber*\label{eq:split_nat1}\\
	\quad
	\boldsymbol{\Theta}_{\mathrm{post}}
	&= -\Bigl(\sum_{n=1}^N \mathbf{C}_n + \mathbf{D}\Bigr),
	\label{eq:split_nat2}
\end{IEEEeqnarray}
where \(\mathbf{b}_n\) and \(\mathbf{C}_n\) are given by
\begin{IEEEeqnarray}{Cl}
	\mathbf{P}_{1n}\,\mathbf{b}_n
	&=
	\begin{bmatrix}
		\sigma_z^{-2}\,\mathbf{h}_n^H\mathbf{y}\\[4pt]
		\mathbf{0}
	\end{bmatrix},\IEEEyesnumber\IEEEyessubnumber*\\
	\quad
	\mathbf{P}_{1n}\,\mathbf{C}_n\,\mathbf{P}_{1n}^H
	&=
	\begin{bmatrix}
		0 & \bar{\mathbf{k}}_n^H\\[4pt]
		\bar{\mathbf{k}}_n & \mathbf{0}_{N-1}
	\end{bmatrix}.
\end{IEEEeqnarray}
As we can see, each $\mathbf{b}_n$ is a scaled element from the matched filter, and each $\mathbf{C}_n$ is a negative semi-definite matrix and looks like a cross-splitting of $\mathbf{K}$. 

The \(n\)-th auxiliary manifold is defined by
\begin{IEEEeqnarray}{Cl}
	\label{eq:Ams}
	\mathcal{M}_n:
	\quad
	p_n(\mathbf{x}) \propto 
	\exp \bigl\{\mathbf{x}^H\boldsymbol{\theta}_n
	+\boldsymbol{\theta}_n^H\mathbf{x}
	+\mathbf{x}^H\boldsymbol{\Theta}_n\mathbf{x}\bigr\},
\end{IEEEeqnarray}
with natural parameters
\begin{IEEEeqnarray}{Cl}
	\label{eq: natural_para_AM}
	\boldsymbol{\theta}_n &= \mathbf{b}_n + \boldsymbol{\lambda}_n,\IEEEyesnumber\IEEEyessubnumber\\
	\boldsymbol{\Theta}_n &= -\bigl(\mathbf{C}_n + \mathbf{D} + \boldsymbol{\Lambda}_n\bigr),\IEEEyessubnumber
\end{IEEEeqnarray}
where \(\boldsymbol{\lambda}_n\) and diagonal matrix  \(\boldsymbol{\Lambda}_n\) are free variables. 
The aim is to find  fixed \(\boldsymbol{\lambda}_n\) and   \(\boldsymbol{\Lambda}_n\) to replace \(\sum_{m\neq n}\mathbf{b}_m\) and \(\sum_{m\neq n}\mathbf{C}_m\) in $ p(\mathbf{x}|\mathbf{y}) $ while keeping the mean and diagonal of covariance matrix unchanged.  


Similarly, the natural parameters of the objective manifold  
\begin{IEEEeqnarray}{Cl}
	\mathcal{M}_0:
	\quad
	p_0(\mathbf{x}) \propto 
	\exp \bigl\{\mathbf{x}^H\boldsymbol{\theta}_0 
	+ \boldsymbol{\theta}_0^H\mathbf{x}
	+\mathbf{x}^H\boldsymbol{\Theta}_0\mathbf{x}\bigr\},
\end{IEEEeqnarray}
is defined as
\begin{IEEEeqnarray}{Cl}
	\label{eq: natural_param_OBM}
	\boldsymbol{\theta}_0 &= \boldsymbol{\lambda}_0,\IEEEyesnumber\IEEEyessubnumber\\
	\boldsymbol{\Theta}_0 &= -\bigl(\boldsymbol{\Lambda}_0 + \mathbf{D}\bigr),\IEEEyessubnumber
\end{IEEEeqnarray}
where \(\boldsymbol{\lambda}_0\) and diagonal matrix \(\boldsymbol{\Lambda}_0\) are free variables. The aim is to find fixed  \(\boldsymbol{\lambda}_0\) and  \(\boldsymbol{\Lambda}_0\) to  replace \(\sum_m\mathbf{b}_m\) and \(\sum_m\mathbf{C}_m\) in $ p(\mathbf{x}|\mathbf{y}) $ while keeping the mean and diagonal of covariance matrix unchanged.  

To guarantee that we can find fixed \(\boldsymbol{\lambda}_n\),   \(\boldsymbol{\Lambda}_n\),  \(\boldsymbol{\lambda}_0\) and  \(\boldsymbol{\Lambda}_0\) as we want,  the \(e\)- and \(m\)-conditions are introduced in the IG framework.  The \(e\)-condition in the natural parameter domain is
\begin{equation}
	\label{eq:e-condition-natural}
	\sum_{n=1}^N (\boldsymbol{\theta}_n,\boldsymbol{\Theta}_n)
	+ (1-N)\,(\boldsymbol{\theta}_0,\boldsymbol{\Theta}_0)
	 = (\boldsymbol{\theta}_{\mathrm{post}},\boldsymbol{\Theta}_{\mathrm{post}}),
\end{equation}
which is equivalent to
\begin{IEEEeqnarray}{Cl}
	\label{eq:e-condition-exp}
	\sum_{n=1}^N (\boldsymbol{\lambda}_n,\boldsymbol{\Lambda}_n)
	+ (1-N)\,(\boldsymbol{\lambda}_0,\boldsymbol{\Lambda}_0)
	\;=\;\mathbf{0}.
\end{IEEEeqnarray}
The \(m\)-condition in the expectation parameter domain is 
\begin{IEEEeqnarray}{Cl}
	\label{eq: m-condition}
	(\boldsymbol{\mu}_n,\mathbf{I}\odot\boldsymbol{\Sigma}_n)
	= (\boldsymbol{\mu}_0,\mathbf{I}\odot\boldsymbol{\Sigma}_0),
	\quad
	\forall n=1,\dots,N.
\end{IEEEeqnarray}
Here, \(\boldsymbol{\mu}_n,\boldsymbol{\Sigma}_n\) and \(\boldsymbol{\mu}_0,\boldsymbol{\Sigma}_0\) denote the mean and covariance of the \(n\)-th auxiliary PDF and the objective PDF, respectively.  The $m$-condition ensures all auxiliary PDFs and objective PDF have the same mean and diagonal of covariance matrix. 
The $e$-condition ensures these means and diagonal of covariance matrices are good approximates of that of the \emph{a posteriori} PDF.

\subsection{CS-IGA for signal detection}

The $e$-condition always holds in the iterative process. 
In the following, we introduce several steps to make the $m$-condition hold.

First, we introduce the calculation of  \(\boldsymbol{\mu}_n,\boldsymbol{\Sigma}_n\).
For each auxiliary defined in \eqref{eq: natural_para_AM}, we denote the $n$‑th components of 
\(
(\boldsymbol{\lambda}_n,\boldsymbol{\Lambda}_n,\mathbf{D},\boldsymbol{\mu}_n)
\)
by 
\(
(\lambda_n,\Lambda_n,d_n,\mu_n),
\)
and collect the remaining $(N -1)$ entries into
\(
(\bar{\boldsymbol{\lambda}}_n,\bar{\boldsymbol{\Lambda}}_n,\bar{\mathbf{D}}_n,\bar{\boldsymbol{\mu}}_n).
\)
With this notation, the following theorem gives low complexity calculations for the expectation parameters of the $n$‑th auxiliary manifold.

\begin{theorem} \label{theorem 1}
	For the $n$-th auxiliary manifold with natural parameters $(\boldsymbol{\theta}_n,\boldsymbol{\Theta}_n)$, the expectation parameters are calculated as
	\begin{IEEEeqnarray}{Cl}
		\label{eq: mu_q and Sigma_q}
		&\mathbf{P}_{1n}\boldsymbol{\mu}_n=\begin{bmatrix}
			\mu_n\\
			\bar{\boldsymbol{\mu}}_n
		\end{bmatrix} = \left[\begin{array}{c}
			\frac{r_n}{\sigma_z^2}\mathbf{h}_n^H\mathbf{y}+r_n\lambda_n-r_nv_n\\
			-\mu_n\check{\boldsymbol{\Lambda}}_n\bar{\mathbf{k}}_n+\check{\boldsymbol{\Lambda}}_n\bar{\boldsymbol{\lambda}}_{n}
		\end{array}\right],
		\IEEEyesnumber\IEEEyessubnumber\label{eq: mu_q}\\
		&\mathbf{P}_{1n}\boldsymbol{\Sigma}_n\mathbf{P}_{1n}^H =  \left[\begin{array}{cc}
			r_n & \mathbf{m}_n^H \\ 
			\mathbf{m}_n  & \check{\boldsymbol{\Lambda}}_n+{r_n}\check{\boldsymbol{\Lambda}}_n\bar{\mathbf{k}}_n\bar{\mathbf{k}}_n^H\check{\boldsymbol{\Lambda}}_n
		\end{array}\right],
		\label{eq: Sigma_q}\IEEEyessubnumber
	\end{IEEEeqnarray}
	where $\check{\boldsymbol{\Lambda}}_n$, $r_n$, $\mathbf{m}_n$ and $v_n$ are given by
	\begin{IEEEeqnarray}{Cl}
		&\check{\boldsymbol{\Lambda}}_n = (\bar{\boldsymbol{\Lambda}}_{ n}+\bar{\mathbf{D}}_{ n})^{-1}, 
		\IEEEyesnumber\IEEEyessubnumber\\
		& r_n = (\Lambda_n+d_n- \bar{\mathbf{k}}_n^H\check{\boldsymbol{\Lambda}}_n\bar{\mathbf{k}}_n)^{-1},
		\IEEEyessubnumber\\
		& \mathbf{m}_n = -{r_n}\check{\boldsymbol{\Lambda}}_n\bar{\mathbf{k}}_n,~ v_n=\bar{\mathbf{k}}_n^H\check{\boldsymbol{\Lambda}}_n\bar{\boldsymbol{\lambda}}_{n}.
		\IEEEyessubnumber
	\end{IEEEeqnarray}
\end{theorem}
\begin{proof}
	See Appendix \ref{Appendix A}.
\end{proof}


Then, the $m$-projection of the $n$-th auxiliary PDF to the objective manifold is obtained from \eqref{eq:m_projection} as
\begin{IEEEeqnarray}{Cl}
	\label{eq: m-projection}
	\mathbf{P}_{1n}\boldsymbol{\Theta}^0_n\mathbf{P}_{1n}^H &= -\left(\mathbf{I}\odot\mathbf{P}_{1n}\boldsymbol{\Sigma}_n\mathbf{P}_{1n}^H\right)^{-1},
	\label{eq: m-projection_Theta}\IEEEyesnumber\IEEEyessubnumber\\
	\mathbf{P}_{1n}\boldsymbol{\theta}^0_n &= -\mathbf{P}_{1n}\boldsymbol{\Theta}_n^0\boldsymbol{\mu}_n.
	\IEEEyessubnumber\label{eq: m-projection_theta}
\end{IEEEeqnarray}
By substituting \eqref{eq: Sigma_q} into \eqref{eq: m-projection_Theta}, $\boldsymbol{\Theta}_n^0$ can be obtained by
\begin{IEEEeqnarray}{Cl}
	\mathbf{P}_{1n}\boldsymbol{\Theta}_n^0\mathbf{P}_{1n}^H&=
	\begin{bmatrix}
		r_n&\mathbf{0}\\
		\mathbf{0}&\check{\boldsymbol{\Lambda}}_n+r_n\check{\boldsymbol{\Lambda}}_n\bar{\mathbf{L}}_n\check{\boldsymbol{\Lambda}}_n
	\end{bmatrix}^{-1},
\end{IEEEeqnarray}
where  $\bar{\mathbf{L}}_n = \mathbf{I}\odot(\bar{\mathbf{k}}_n\bar{\mathbf{k}}_n^H)$ is diagonal. For simplicity, we define $\boldsymbol{\Phi}_n = \check{\boldsymbol{\Lambda}}_n+r_n\check{\boldsymbol{\Lambda}}_n\bar{\mathbf{L}}_n\check{\boldsymbol{\Lambda}}_n$, and its inversion can be simplified as
\begin{IEEEeqnarray}{Cl}
	\boldsymbol{\Phi}_n^{-1} &= (\check{\boldsymbol{\Lambda}}_n+r_n\check{\boldsymbol{\Lambda}}_n\bar{\mathbf{L}}_n\check{\boldsymbol{\Lambda}}_n)^{-1} \notag\\
	&= \check{\boldsymbol{\Lambda}}_n^{-1}(\mathbf{I}+r_n\bar{\mathbf{L}}_n\check{\boldsymbol{\Lambda}}_n)^{-1} \notag\\
	&= \check{\boldsymbol{\Lambda}}_n^{-1}- \check{\boldsymbol{\Lambda}}_n^{-1}\left(\mathbf{I}+r_n\bar{\mathbf{L}}_n\check{\boldsymbol{\Lambda}}_n\right)^{-1}r_n\check{\boldsymbol{\Lambda}}_n\bar{\mathbf{L}}_n \notag\\
	&= \check{\boldsymbol{\Lambda}}_n^{-1}-r_n\bar{\mathbf{L}}_n\left(\mathbf{I}+r_n\bar{\mathbf{L}}_n\check{\boldsymbol{\Lambda}}_n\right)^{-1}.
\end{IEEEeqnarray}
We denote  $\mathbf{R}_n=(\mathbf{I}+r_n\bar{\mathbf{L}}_n\check{\boldsymbol{\Lambda}}_n)^{-1}$. Then $\boldsymbol{\Phi}_n^{-1} = \check{\boldsymbol{\Lambda}}_n^{-1}-r_n\bar{\mathbf L}_n\mathbf{R}_n$, and $\mathbf{P}_{1n}\boldsymbol{\Theta}_n^0\mathbf{P}_{1n}^H$ is rewritten as
\begin{IEEEeqnarray}{Cl}
	\label{eq:Theta_n_0}
	\mathbf{P}_{1n}\boldsymbol{\Theta}_n^0\mathbf{P}_{1n}^H &= \begin{bmatrix}
		r_n^{-1}&\mathbf{0}\\
		\mathbf{0}& \check{\boldsymbol{\Lambda}}_n^{-1}-r_n\bar{\mathbf{L}}_n\mathbf{R}_n
	\end{bmatrix}.
\end{IEEEeqnarray}
Next, we calculate $\mathbf{P}_{1n}\boldsymbol{\theta}_n^0$. By substituting \eqref{eq: mu_q} and \eqref{eq:Theta_n_0} into \eqref{eq: m-projection_theta}, $\mathbf{P}_{1n}\boldsymbol{\theta}_n^0$ can be rewritten as
\begin{IEEEeqnarray}{Cl}
	\mathbf{P}_{1n}\boldsymbol{\theta}_n^0&= \mathbf{P}_{1n}\begin{bmatrix}
		r_n^{-1}&\mathbf{0}\\
		\mathbf{0}&\check{\boldsymbol{\Lambda}}^{-1}_n+r_n\bar{\mathbf{L}}_n\mathbf{R}_n
	\end{bmatrix} \notag\\
	&~~~~~\begin{bmatrix}
		r_n\frac{1}{\sigma_z^{2}}\mathbf{h}_n^H\mathbf{y}+r_n\lambda_n-r_nv_n\\
		-\mu_n\check{\boldsymbol{\Lambda}}_n\bar{\mathbf{k}}_n+\check{\boldsymbol{\Lambda}}_n\bar{\boldsymbol{\lambda}}_n
	\end{bmatrix}.
	\label{eq:thetan0_1}
\end{IEEEeqnarray}
Then, the $n$-th element of $\boldsymbol{\theta}_n^0$ is calculated as
\begin{IEEEeqnarray}{Cl}
	\theta_n^0 &=  \frac{1}{\sigma_z^2}\mathbf{h}_n^H\mathbf{y}+\lambda_{ n}-v_n.
	\label{eq:thetan0_0}
\end{IEEEeqnarray}
The remain elements $\bar{\boldsymbol{\theta}}_n^0$ are written as
\begin{IEEEeqnarray}{Cl}
	\bar{\boldsymbol{\theta}}_n^0 &= \left(\check{\boldsymbol{\Lambda}}_n^{-1}-r_n\bar{\mathbf{L}}_n\mathbf{R}_n\right)\left(-\mu_n\check{\boldsymbol{\Lambda}}_n\bar{\mathbf{k}}_n+\check{\boldsymbol{\Lambda}}_n\bar{\boldsymbol{\lambda}}_n\right)\notag\\
	&=\left(\mathbf{I}-r_n\bar{\mathbf{L}}_n\mathbf{R}_n\check{\boldsymbol{\Lambda}}_n\right)(\bar{\boldsymbol{\lambda}}_n-\mu_n\bar{\mathbf{k}}_n).
	\label{eq:thetan0_bar}
\end{IEEEeqnarray}
We now analyze the relationship between  $\left(\mathbf{I}-r_n\bar{\mathbf{L}}_n\mathbf{R}_n\check{\boldsymbol{\Lambda}}_n\right)$ and $\mathbf{R}_n$. For simplicity, we define $\boldsymbol{\Gamma}= r_n\bar{\mathbf{L}}_n\check{\boldsymbol{\Lambda}}_n$, then $\mathbf{R}_n = \left(\mathbf{I}+\boldsymbol{\Gamma}\right)^{-1}$. We then have
\begin{IEEEeqnarray}{Cl}
	\label{eq:Gamma1}
	\mathbf{I}-\boldsymbol{\Gamma}\mathbf{R}_n &= \mathbf{I}-\boldsymbol{\Gamma}(\mathbf{I}+\boldsymbol{\Gamma})^{-1}.
\end{IEEEeqnarray}
By factoring the identity matrix $\mathbf{I}$ as $(\mathbf{I}+\boldsymbol{\Gamma})(\mathbf{I}+\boldsymbol{\Gamma})^{-1}$, \eqref{eq:Gamma1} can be rewritten as
\begin{IEEEeqnarray}{Cl}
	\mathbf{I}-\boldsymbol{\Gamma}\mathbf{R}_n &= [(\mathbf{I}+\boldsymbol{\Gamma})-\boldsymbol{\Gamma}](\mathbf{I}+\boldsymbol{\Gamma})^{-1} \notag\\
	&=(\mathbf{I}+\boldsymbol{\Gamma})^{-1}=\mathbf{R}_n.
\end{IEEEeqnarray}
Thus $\bar{\boldsymbol{\theta}}_n^0$ can be rewritten as
\begin{IEEEeqnarray}{Cl}
	\bar{\boldsymbol{\theta}}_n^0 &= \mathbf{R}_n(\bar{\boldsymbol{\lambda}}_n-\mu_n\bar{\mathbf{k}}_n). 
	\label{eq:thetan0_bar_simplified}
\end{IEEEeqnarray}
By combining \eqref{eq:thetan0_0} and \eqref{eq:thetan0_bar_simplified}, \eqref{eq:thetan0_1} can be rewritten as
\begin{IEEEeqnarray}{Cl}
	\mathbf{P}_{1n}\boldsymbol{\theta}_n^0 &= \begin{bmatrix}
		\frac{1}{\sigma_z^2}\mathbf{h}_n^H\mathbf{y}+\lambda_n-v_n\\
		\mathbf{R}_n(\bar{\boldsymbol{\lambda}}_n-\mu_n\bar{\mathbf{k}}_n)
	\end{bmatrix}.
	\label{eq:theta_n0}
\end{IEEEeqnarray}

Next, we calculate the beliefs that can replace $\mathbf{C}_n$ and $\mathbf{b}_n$. By observing the difference from $(\boldsymbol{\theta}_n^0,\boldsymbol{\Theta}_n^0)$ and $(\boldsymbol{\lambda}_n,\boldsymbol{\Lambda}_n)$, the beliefs corresponding to $\mathbf{C}_n$ and $\mathbf{b}_n$ are expressed as
\begin{IEEEeqnarray}{Cl}
	\label{eq: belief1}
	\boldsymbol{\Xi}_n &= {-\boldsymbol{\Theta}_n^0 - \mathbf{D} - \boldsymbol{\Lambda}_n},
	\IEEEyesnumber\IEEEyessubnumber\\
	\boldsymbol{\xi}_n &= \boldsymbol{\theta}_n^0 - \boldsymbol{\lambda}_n.
	\IEEEyessubnumber
\end{IEEEeqnarray}
Substituting \eqref{eq:Theta_n_0} and \eqref{eq:theta_n0} into \eqref{eq: belief1}, the beliefs can be rewritten as
\begin{IEEEeqnarray}{cl}
	\label{eq: beliefs 2}
	&{\left[\begin{array}{c}
			\Xi_n\\
			{\bar{\boldsymbol{\Xi}}}_{n}
		\end{array}\right]} = \left[\begin{array}{c}-\bar{\mathbf{k}}_n^H\check{\boldsymbol{\Lambda}}_n\bar{\mathbf{k}}_n\\ 
		-{r_n}\bar{\mathbf{L}}_n\mathbf{R}_n
	\end{array}\right],
	\IEEEyesnumber\IEEEyessubnumber\\
	 &\left[\begin{array}{c}
		\xi_n\\
		{\bar{\boldsymbol{\xi}}}_{n}
	\end{array}\right]=  \left[\begin{array}{c}
		\sigma_z^{-2}\mathbf{h}_n^H\mathbf{y}-v_n\\ 
		\mathbf{R}_n(\bar{\boldsymbol{\lambda}}_{ n}-\mu_n\bar{\mathbf{k}}_n)-\bar{\boldsymbol{\lambda}}_{n}
	\end{array}\right].
	\IEEEyessubnumber
\end{IEEEeqnarray}
Notice that $\Xi_n$ can also be written as $\text{tr}(\bar{\mathbf{L}}_n\check{\boldsymbol{\Lambda}}_n)$, where $\bar{\mathbf{L}}_n\check{\boldsymbol{\Lambda}}_n$ has been calculated in $\mathbf{R}_n$, \textit{i.e.}, the calculation of $\Xi_n$ is straight forward. 

Then, the parameters are updated as
\begin{IEEEeqnarray}{lClC}
	\label{eq: para update}
	&{\boldsymbol{\lambda}}_0^{t+1} = \sum_{ n}\boldsymbol{\xi}_n^t,\quad&{\boldsymbol{\Lambda}}_0^{t+1} = \sum_{n}\boldsymbol{\Xi}_n^t,
	\IEEEyesnumber\IEEEyessubnumber\\
	&\boldsymbol{\lambda}_n^{t+1} = \boldsymbol{\lambda}_0^{t+1}-\boldsymbol{\xi}_n^t,\quad&\boldsymbol{\Lambda}_n^{t+1} =\boldsymbol{\Lambda}_0^{t+1}-\boldsymbol{\Xi}_n^t,
	\IEEEyessubnumber
\end{IEEEeqnarray}
where $t$ is the iteration number. It's easy to verify that the $e$-condition always holds.

According to \eqref{eq: belief1} and \eqref{eq: para update}, when the iteration converges, the $m$-projection of the $n$-th auxiliary PDF satisfies
\begin{IEEEeqnarray}{Cl}
	\label{eq: converge point}
	&\boldsymbol{\theta}_n^0 = \boldsymbol{\lambda}_0,
	\IEEEyesnumber\IEEEyessubnumber\\
	&\boldsymbol{\Theta}_n^0 = \boldsymbol{\Theta}_0 =-(\boldsymbol{\Lambda}_0+\mathbf{D}) ,
	\IEEEyessubnumber
\end{IEEEeqnarray}
which makes the $m$-condition hold.

Then, the output mean and variance are easy to obtain by
\begin{IEEEeqnarray}{Cl}
	\label{eq:out mean and variance}
	&\hat{\boldsymbol{\mu}} = -\boldsymbol{\Theta}_0^{-1}\boldsymbol{\theta}_0 = (\boldsymbol{\Lambda}_0+\mathbf{D})^{-1}\boldsymbol{\lambda}_0,
	\IEEEyesnumber\IEEEyessubnumber\\
	&\hat{\boldsymbol{\Sigma}} = -\boldsymbol{\Theta}_0^{-1}=(\boldsymbol{\Lambda}_0+\mathbf{D})^{-1}.
	\IEEEyessubnumber
\end{IEEEeqnarray}
Because of the $e$-condition and $m$-condition, the output mean and variance provide accurate approximations of the original values.
Since $\boldsymbol{\Theta}_0$ is a diagonal matrix, its inversion requires significantly lower computational complexity.

In practice, we might need the damping factor in the iterations to make the algorithm converge. In such case, the parameters are updated as
\begin{IEEEeqnarray}{Cl}
	\label{eq: para update damping}
	&{\boldsymbol{\lambda}}_0^{t+1} = (1-\alpha)\boldsymbol{\lambda}_0^t+\alpha\sum_{ n}\boldsymbol{\xi}_n^t,
	\IEEEyesnumber\IEEEyessubnumber\\
	&{\boldsymbol{\Lambda}}_0^{t+1} = (1-\alpha)\boldsymbol{\Lambda}_0^t+\alpha\sum_{n}\boldsymbol{\Xi}_n^t,
	\IEEEyessubnumber\\
	&\boldsymbol{\lambda}_n^{t+1} = (1-\alpha)\boldsymbol{\lambda}_n^t+\alpha\sum_{m\neq n}\boldsymbol{\xi}_m^t,
	\IEEEyessubnumber\\
	&\boldsymbol{\Lambda}_n^{t+1} =(1-\alpha)\boldsymbol{\Lambda}_n^{t}+\alpha\sum_{m\neq n}\boldsymbol{\Xi}_m^t,
	\IEEEyessubnumber
\end{IEEEeqnarray}
where $\alpha$ is the damping factor. 

We summarize the CS-IGA in Algorithm \ref{algorithmn_CSIGA}.
\begin{algorithm}[!hbt]
	\KwData{Channel matrix $\mathbf{H}$, received signal $\mathbf{y}$, damping factor $\alpha$, noise power $\sigma_z^2$, maximum iteration number $T$.}
	\KwResult{The mean and variance of the detected signal}
	\SetKwProg{Fn}{Function}{:}{}
	\Fn{CS-IGA ($\mathbf{H},\mathbf{y},\alpha,\sigma_z^2,T$)}{
		$\boldsymbol{\lambda} \gets\mathbf{0}, \boldsymbol{\Lambda} \gets \mathbf{-1}$\;
		
		\For{$t \gets 1$ \textbf{to} $T$}{
			Calculate ${{\mu}}_n$ for all auxiliary manifolds as \eqref{eq: mu_q} \;
			Calculate the beliefs for all auxiliary manifolds as \eqref{eq: beliefs 2}\;
			Update parameters as \eqref{eq: para update damping}\;
		}
		Calculate $\hat{\boldsymbol{\mu}}$ and $\hat{\boldsymbol{\Sigma}}$ as \eqref{eq:out mean and variance}\;
		\KwRet{$\hat{\boldsymbol{\mu}}, \hat{\boldsymbol{\Sigma}}$}\;
	}
	\caption{CS-IGA for Signal Detection}
	\label{algorithmn_CSIGA}
\end{algorithm}

\begin{figure*}[!t]
	\normalsize
	\begin{IEEEeqnarray}{Cl}
		\mathcal{H}_e &= \left\{p(\mathbf{x};\boldsymbol{\theta},\boldsymbol{\Theta}) \Big| (\boldsymbol{\theta},\boldsymbol{\Theta})
		=\sum_{n=1}^{N} \alpha_i (\boldsymbol{\theta}_{n},\boldsymbol{\Theta}_n)
		+(1-\sum_{n=1}^{N} \alpha_i) (\boldsymbol{\theta}_0 ,\boldsymbol{\Theta}_0)
		\right\},\IEEEyesnumber\IEEEyessubnumber*\label{eq:e-hyperplane}\\
		\mathcal{H}_m &= \left\{p(\mathbf{x};\boldsymbol{\mu},\boldsymbol{\Sigma})\;\Big|\;(\boldsymbol{\mu},\mathbf{I}\odot\boldsymbol{\Sigma})	=(\boldsymbol{\mu}_n,\mathbf{I}\odot\boldsymbol{\Sigma}_n)
				=(\boldsymbol{\mu}_0,\mathbf{I}\odot\boldsymbol{\Sigma}_0),
		\;\forall n\in\mathcal{Z}_N^+
		\right\}.\label{eq:m-hyperplane}
	\end{IEEEeqnarray}
	\hrulefill
\end{figure*}
\subsection{Information Geometrical Interpretation}
The iterative process of the CS-IGA algorithm has an elegant geometric interpretation. The convergence can be visualized as a projection process on a statistical manifold, where multiple approximate distributions (the auxiliary points) are forced to align with a simplified target distribution (the objective point). Recall that the original point lies at \eqref{eq:linear posterior}.  By construction, the objective point and auxiliary point satisfy the \(e\)‑condition \eqref{eq:e-condition-natural} at every iteration.  Consequently, the objective point, the original point, and the auxiliary points all lie on the same \(e\)‑hyperplane \(\mathcal{H}_e\) defined in \eqref{eq:e-hyperplane}.  Here \((\boldsymbol{\theta}_0,\boldsymbol{\Theta}_0)\) and \((\boldsymbol{\theta}_n,\boldsymbol{\Theta}_n)\) vary during the algorithm, while \((\boldsymbol{\theta}_{\mathrm{post}},\boldsymbol{\Theta}_{\mathrm{post}})\) remains fixed.

On the dual side, the \(m\)‑hyperplane \(\mathcal{H}_m\) in \eqref{eq:m-hyperplane} constrains all points except the original point to share the same expectation parameters-namely, the \emph{a posteriori} mean and the diagonal of the \emph{a posteriori} covariance.  At initialization, the auxiliary points generally do not satisfy \(\mathcal{H}_m\).  As CS‑IGA iterates, the auxiliary points are \(m\)‑projected onto the objective manifold via \eqref{eq: m-projection}, and upon convergence both the objective and auxiliary points lie on \(\mathcal{H}_m\).  

Since  \(\mathcal{H}_e\) and \(\mathcal{H}_m\) shares the objective point and the auxiliary points, the original point lie on \(\mathcal{H}_e\) will be close to the \(\mathcal{H}_m\). 
Thus, the mean and diagonal of the \emph{a posteriori} covariance of points on \(\mathcal{H}_m\) is a good appoximate of the original point. 
In certain case, we can even prove that mean of points on \(\mathcal{H}_m\) is equal to that of the original point.


\subsection{Fixed Point and Computational Complexity}
In this subsection, we first prove the mean at the fixed point of CS‑IGA is equal to the LMMSE detection. Then, we analyze the computational complexity of the CS‑IGA and compare it with the Bayes‑optimal AMP and the traditional IGA.

When the algorithm converges, both the \(e\)‑condition and \(m\)‑condition hold, which leads to
\begin{IEEEeqnarray}{Cl}
	&(\boldsymbol{\mu}_n^{\bullet},\mathbf{I}\odot\boldsymbol{\Sigma}_n^{\bullet})  =(\boldsymbol{\mu}_0^{\bullet},\boldsymbol{\Sigma}_0^{\bullet}), 
	\IEEEyesnumber\IEEEyessubnumber\label{eq: fixed point mean}\\
	&\sum_{n=1}^{N}(\boldsymbol{\lambda}_n^{\bullet},\boldsymbol{\Lambda}_n^{\bullet})+(1-N)(\boldsymbol{\lambda}_0^{\bullet},\boldsymbol{\Lambda}_0^{\bullet})=\mathbf{0}.
	\IEEEyessubnumber\label{eq: fixed point var}
\end{IEEEeqnarray}
The following theorem shows that the fixed point of  the CS‑IGA is equivalent to the LMMSE estimation of \(\mathbf{x}\).

\begin{theorem} 
	At the fixed point of the CS‑IGA, the mean \(\boldsymbol{\mu}^{\bullet}\) is equal to the LMMSE estimation or the \emph{a posteriori} mean in \eqref{eq:posterior mean}.
\end{theorem}

\begin{proof}
	See Appendix \ref{Appendix B}.
\end{proof}

Next, we analyze the computational complexity of the CS‑IGA. Before the iterations begin, the CS‑IGA computes \(\mathbf{H}^H\mathbf{H}\), which is commonly done in user scheduling \cite{schedule}. The CS-IGA also calculate the matched filter result \(\mathbf{H}^H\mathbf{y}\), which is shared across methods. We therefore omit these from the iteration complexity. In each iteration, CS‑IGA computes \(\boldsymbol{\Xi}_n\) and \(\boldsymbol{\xi}_n\) for each AM. Since both \(\bar{\mathbf{L}}_n\) and \(\mathbf{R}_n\) are diagonal, each update costs \(\mathcal{O}(N)\). Over \(N\) AMs and \(T\) iterations, the total complexity is \(\mathcal{O}(T N^2)\). For small \(T\), this complexity is significantly less than the \(\mathcal{O}(N^3)\) required for direct LMMSE detection. Compared to traditional IGA \cite{IGChEst} and Bayes‑optimal AMP, each of which has complexity \(\mathcal{O}(T M N)\), CS‑IGA is more efficient whenever \(N < M\), as is typical in practical XL‑MIMO systems.

\section{NCS‑IGA for Non‑Linear Detection}\label{CS_IGA Non-Linear}

In this section, we extend the CS‑IGA framework to the non‑linear setting, yielding the NCS‑IGA algorithm.

\subsection{Algorithm Derivation}

The goal of the NCS‑IGA is to approximate the \emph{a posteriori} marginal distributions \(\{p(x_k|\mathbf{y})\}\).
To derive the NCS‑IGA, we need to add an extra auxiliary manifold that incorporate the discrete priors \(p_{pr,n}(x_n)\).  At each iteration, the algorithm computes the beliefs by projecting auxiliary points onto the objective manifold, exactly as in the  CS‑IGA case except the extra manifold (cf.\ \eqref{eq: m-projection}–\eqref{eq: para update damping}). 
The aim of defining the extra auxiliary manifold is the same as the Gaussian approximation in the EP, but in the IG view. The resulting estimates of \(p (x_k|\mathbf{y})\) are then used  to form soft decisions.


For the \textit{a posteriori} distribution in \eqref{eq:posterior non-liner}, $\mathbf{x}$ is not considered Gaussian, thus \eqref{eq:posterior_mfld} is rewritten as
\begin{IEEEeqnarray}{rCl}
	\label{eq:posterior mf nonlinear1}
	p(\mathbf{x}\mid\mathbf{y})
	&\propto& p_{pr}(\mathbf{x})\exp\!\Bigl\{\mathbf{x}^H\underbrace{\sigma_z^{-2}\mathbf{H}^H\mathbf{y}}_{\boldsymbol{\theta}_{p}}
	\;+\;\underbrace{\mathbf{y}^H\mathbf{H}\,\sigma_z^{-2}}_{\boldsymbol{\theta}_{p}^H}\mathbf{x}\notag\\
	&\;+\;&\mathbf{x}^H\underbrace{[-(\sigma_z^{-2}\mathbf{H}^H\mathbf{H})]}_{\boldsymbol{\Theta}_{p}}\mathbf{x}\Bigr\}.
\end{IEEEeqnarray}

To use the result of the CS-IGA, we rewrite \eqref{eq:posterior mf nonlinear1} as
\begin{IEEEeqnarray}{Cl}
	\label{eq:posterior mf nonlinear2}
	p(\mathbf{x}\mid\mathbf{y})\propto p_{pr}(\mathbf{x})p(\mathbf{x};\boldsymbol{\theta}_p,\boldsymbol{\Theta}_p),
\end{IEEEeqnarray}
where $p(\mathbf{x};\boldsymbol{\theta}_p,\boldsymbol{\Theta}_p)$ is denoted by
\begin{IEEEeqnarray}{Cl}
	p(\mathbf{x};\boldsymbol{\theta}_p,\boldsymbol{\Theta}_p) \propto \exp\left\{\mathbf{x}^H\boldsymbol{\theta}_p+\boldsymbol{\theta}_p^H\mathbf{x}+\mathbf{x}^H\boldsymbol{\Theta}_p\mathbf{x}\right\}.
\end{IEEEeqnarray}

 
We still apply the splitting in \eqref{eq:splitting}, but with $\boldsymbol{\theta}_{\mathrm{post}}$ and $\boldsymbol{\Theta}_{\mathrm{post}}$ changing to $\boldsymbol{\theta}_{p}$ and $\boldsymbol{\Theta}_{p}$ and $ \mathbf{K}=\sigma_z^{-2}\mathbf{H}^H\mathbf{H}$. The definitions of the auxiliary manifolds and the objective manifold also keep the same. However, we now need to find  fixed \(\boldsymbol{\lambda}_n\) and   \(\boldsymbol{\Lambda}_n\) to replace \(\sum_{m\neq n}\mathbf{b}_m\) and \(\sum_{m\neq n}\mathbf{C}_m\) and also the prior term in $ p(\mathbf{x}|\mathbf{y}) $  while keeping the mean and diagonal of covariance matrix unchanged.  For \(\boldsymbol{\lambda}_0\) and   \(\boldsymbol{\Lambda}_0\), the situation is similar.

Then, an extra auxiliary PDF is constructed as
\begin{IEEEeqnarray}{Cl}
	\label{eq:extra-PDF}
	p_e(\mathbf{x};\hat{\boldsymbol{\theta}}_0,\hat{\boldsymbol{\Theta}}_0)\propto p_{pr}(\mathbf{x})p(\mathbf{x};\hat{\boldsymbol{\theta}}_0,\hat{\boldsymbol{\Theta}}_0),
\end{IEEEeqnarray}
where $\hat{\boldsymbol{\theta}}_0$ and $\hat{\boldsymbol{\Theta}}_0$ are natural parameters. They are denoted as
\begin{IEEEeqnarray}{Cl} 
	\hat{\boldsymbol{\theta}}_0 &= \hat{\boldsymbol{\lambda}}_0=\sum_{n=1}^{N}\boldsymbol{\xi}_n,\IEEEyesnumber\IEEEyessubnumber\\
	\hat{\boldsymbol{\Theta}}_0 &= -\bigl(\hat{\boldsymbol{\Lambda}}_0 + {\mathbf{D}}\bigr)=-\left(\sum_{n=1}^{N}\boldsymbol{\Xi}_n+{\mathbf{D}}\right),\IEEEyessubnumber
\end{IEEEeqnarray}
where $\hat{\boldsymbol{\Theta}}_0$ is diagonal and ${\mathbf{D}}=\mathbf{I}\odot\mathbf{K}$. The corresponding extra auxiliary manifold is defined as
\begin{IEEEeqnarray}{Cl}
	\label{eq:extra-manifold}
	\mathcal{M}_e = \left\{p_e(\mathbf{x};\hat{\boldsymbol{\theta}}_0,\hat{\boldsymbol{\Theta}}_0)\right\}.
\end{IEEEeqnarray}

Since an extra auxiliary manifold is introduced, the $e$-condition in \eqref{eq:e-condition-natural} is rewritten as
\begin{IEEEeqnarray}{Cl}
	\label{eq:e-condition-natrual-NonLinear}
	\sum_{n=1}^{N}(\boldsymbol{\theta}_n,\boldsymbol{\Theta}_n)+(\hat{\boldsymbol{\theta}}_0,\hat{\boldsymbol{\Theta}}_0)-N(\boldsymbol{\theta}_0,\boldsymbol{\Theta}_0) = (\boldsymbol{\theta}_p,\boldsymbol{\Theta}_p),
\end{IEEEeqnarray}
which is equivalent to
\begin{IEEEeqnarray}{Cl}
	\label{eq:e-condition-exp-NonLinear}
	\sum_{n=1}^{N}(\boldsymbol{\lambda}_n,\boldsymbol{\Lambda}_n)+(\hat{\boldsymbol{\lambda}}_0,\hat{\boldsymbol{\Lambda}}_0)-N(\boldsymbol{\lambda}_0,\boldsymbol{\Lambda}_0) = \mathbf{0}.
\end{IEEEeqnarray}
The $m$-condition considered the extra auxiliary manifold is expressed as
\begin{IEEEeqnarray}{Cl}
	\label{eq: m-condition-NonLinear}
	(\boldsymbol{\mu}_n,\mathbf{I}\odot\boldsymbol{\Sigma}_n) =(\hat{\boldsymbol{\mu}}_0,\mathbf{I}\odot\hat{\boldsymbol{\Sigma}}_0)= (\boldsymbol{\mu}_0,\mathbf{I}\odot\boldsymbol{\Sigma}_0),
\end{IEEEeqnarray}
where $\hat{\boldsymbol{\mu}}_0$ and $\hat{\boldsymbol{\Sigma}}_0$ are the mean and variance of the extra auxiliary manifold, respectively.

The computations of $\boldsymbol{\xi}_n^t, \boldsymbol{\Xi}_n^t$ are the same as those in the CS-IGA steps, we only need to change the definition of $\mathbf{K}$.

The main task remain is then to $m$-project the extra PDF to the objective manifold. To do so, we need to calculate the \emph{a posteriori} probability of each symbol. For the $k$-th symbol, the \emph{a posteriori} distribution is written as
\begin{IEEEeqnarray}{Cl}
	p(x_k\mid\mathbf{y})\propto p_{pr,k}(x_k)p(\mathbf{x};\hat{\boldsymbol{\theta}}_0,\hat{\boldsymbol{\Theta}}_0),
\end{IEEEeqnarray}
where $p(\mathbf{x};\hat{\boldsymbol{\theta}_0},\hat{\boldsymbol{\Theta}}_0)$ is Gaussian, its mean and variance are denoted as
\begin{IEEEeqnarray}{Cl}
	\hat{\boldsymbol{\mu}}_0 &= -\hat{\boldsymbol{\Theta}}_0^{-1}\hat{\boldsymbol{\theta}}_0,\IEEEyesnumber\IEEEyessubnumber\\
	\hat{\boldsymbol{\Sigma}}_0 &= -\hat{\boldsymbol{\Theta}}_0^{-1}.\IEEEyessubnumber 
\end{IEEEeqnarray}
For each constellation point $\{x^{(\ell)}\}_{\ell=1}^L\in\mathbb{X}^L$, the \textit{a posteriori} probability is calculated as
\begin{IEEEeqnarray}{Cl}
	\label{eq:poster eta}
	p(x_k\mid\mathbf{y})_{x_k=x^{(\ell)}} &= p_{pr,k}(x_k)_{x_k=x^{(\ell)}}p(x^{(\ell)};\hat{\boldsymbol{\theta}}_0,\hat{\boldsymbol{\Theta}}_0),\notag\\
	&\propto\exp\left\{\frac{\Vert\hat{\mu}_{0,k}-x^{(\ell)}\Vert^2_2}{\hat{\Sigma}_{0,k}}\right\}.
\end{IEEEeqnarray}
Then, the probability of the $k$-th symbol on the $\ell$-th constellation point is denoted as
\begin{IEEEeqnarray}{Cl}
	\label{eq:symbol probability}
	\eta_{k,\ell} = \frac{\hat{\Sigma}_{0,k}^{-1}\Vert\hat{\mu}_{0,k}-x^{(\ell)}\Vert^2_2}{\sum_{x^{(m)}\in\mathbb{X}^K}\hat{\Sigma}_{0,k}^{-1}\Vert\hat{\mu}_{0,k}-x^{(m)}\Vert^2_2}.
\end{IEEEeqnarray}
For simplicity, we use $\boldsymbol{\eta}_k\in\mathbb{C}^{L\times 1}$ to denote the probability of the $k$-th symbol, where $\boldsymbol{\eta}_k = [\eta_{k,1},\eta_{k,2},\cdots,\eta_{k,L}]^T$ and $\sum_{\ell=1}^{L}{\eta}_{k,\ell}=1$. We then project this discrete probability to a Gaussian distribution whose mean and variance are calculated as
\begin{IEEEeqnarray}{Cl}
	\label{eq:mean and var of extra AM}
	\tilde{{\mu}}_{0,k} = \sum\boldsymbol{\eta}_k\odot\mathbf{c},\IEEEyesnumber\IEEEyessubnumber\label{eq:mean of the extra AM}\\
	\tilde{\Sigma}_{0,k} = \sum\boldsymbol{\eta}_k\odot(\mathbf{c}-\tilde{{\mu}}_{0,k})\odot(\mathbf{c}-\tilde{\mu}_{0,k})^*,\IEEEyessubnumber
\end{IEEEeqnarray}
where $\mathbf{c}$ denotes the complex constellation set. Let $\tilde{\boldsymbol{\mu}}_0 = [\tilde{\mu}_{0,1},\tilde{\mu}_{0,2},\cdots,\tilde{\mu}_{0,N}]^T$ and $\tilde{\boldsymbol{\Sigma}}_0 = [\tilde{\Sigma}_{0,1}, \tilde{\Sigma}_{0,2}, \cdots, \tilde{\Sigma}_{0,N}]^T$ denote the $m$-projected mean and variance of the extra auxiliary PDF, respectively. 

After obtaining the expectation parameters, the natural parameters of the objective PDF can be easily obtained by
\begin{IEEEeqnarray}{Cl}
	\label{eq:natrual param proj extra AM}
	{\boldsymbol{\lambda}}_0 &= \tilde{\boldsymbol{\Sigma}}_0^{-1}\tilde{\boldsymbol{\mu}}_0,\IEEEyesnumber\IEEEyessubnumber*\\
	{\boldsymbol{\Lambda}}_0 &= \tilde{\boldsymbol{\Sigma}}_0^{-1}.
\end{IEEEeqnarray}
By comparing $(\hat{\boldsymbol{\lambda}}_0,\hat{\boldsymbol{\Lambda}}_0)$ and $(\boldsymbol{\theta}_0,\boldsymbol{\Theta}_0)$, the beliefs corresponding to the prior PDF can be easily obtained by
\begin{IEEEeqnarray}{Cl}
	\label{eq:extra beliefs}
	\boldsymbol{\xi}_e &={\boldsymbol{\lambda}}_0-\hat{\boldsymbol{\lambda}}_0,  \IEEEyesnumber\IEEEyessubnumber*\\
	\boldsymbol{\Xi}_e &={\boldsymbol{\Lambda}}_0 -\hat{\boldsymbol{\Lambda}}_0 - {\mathbf{D}}.
\end{IEEEeqnarray}
Then, the parameters can be updated as
\begin{IEEEeqnarray}{Cl}
	\label{eq:para_upd NonLinear}
	&\hat{\boldsymbol{\lambda}}_0^{t+1} = \sum_{n=1}^{N}\boldsymbol{\xi}_n^t,\IEEEyesnumber\IEEEyessubnumber*\\
	&\hat{\boldsymbol{\Lambda}}_0^{t+1} = \sum_{n=1}^{N}\boldsymbol{\Xi}_n^{t},\\
	&\boldsymbol{\lambda}_n^{t+1} = \hat{\boldsymbol{\lambda}}_0^{t+1}+\boldsymbol{\xi}_e^t-\boldsymbol{\xi}_n^t,\\ &{\boldsymbol{\Lambda}}_n^{t+1} =  \hat{\boldsymbol{\Lambda}}_0^{t+1}+\boldsymbol{\Xi}_e^t-\boldsymbol{\Xi}_n^t,
\end{IEEEeqnarray}
where $t$ is the iteration number. Sometimes the algorithm will need the damping factor to converge, which can be expressed as
\begin{IEEEeqnarray}{Cl}
	\label{eq:para_upd NonLinear damp}
	&\hat{\boldsymbol{\lambda}}_0^{t+1} = \alpha\sum_{n=1}^{N}\boldsymbol{\xi}_n^t+(1-\alpha)\hat{\boldsymbol{\lambda}}_0^{t},\IEEEyesnumber\IEEEyessubnumber*\\
	&\hat{\boldsymbol{\Lambda}}_0^{t+1} = \alpha\sum_{n=1}^{N}\boldsymbol{\Xi}_n^{t}+(1-\alpha)\hat{\boldsymbol{\Lambda}}_0^{t},\\
	&\boldsymbol{\lambda}_n^{t+1} = \alpha(\sum_{m=1}^{N}\boldsymbol{\xi}_m^t+\boldsymbol{\xi}_e^t-\boldsymbol{\xi}_n^t)+(1-\alpha)\boldsymbol{\lambda}_n^{t},\\ &{\boldsymbol{\Lambda}}_n^{t+1} =  \alpha(\sum_{m=1}^{N}\boldsymbol{\Xi}_m^t+\boldsymbol{\Xi}_e^t-\boldsymbol{\Xi}_n^t)+(1-\alpha){\boldsymbol{\Lambda}}_n^{t},
\end{IEEEeqnarray}
where $\alpha$ is the damping factor.

In NCS‑IGA, we obtain, for each user \(k\), a posterior symbol‐wise probability vector $\boldsymbol{\eta}_k$ where
\[
\eta_{k,\ell}
= p\bigl(x_k = x^{(\ell)} \mid \mathbf{y}\bigr).
\]
Each symbol \(x^{(\ell)}\) is corresponding to \(B=\log_2L\) bits \(\bigl[b_1(x^{(\ell)}),\dots,b_B(x^{(\ell)})\bigr]\).

We now convert these symbol‐level probabilities into bit‐level log‐likelihood ratios (LLRs) for soft demodulation.  Let
\[
\mathcal{S}_i^b = \bigl\{\ell : b_i(x^{(\ell)})=b\bigr\},
\qquad
b\in\{0,1\},\quad i=1,\dots,B,
\]
be the subsets of constellation indices whose \(i\)-th bit equals \(b\).  Then the bit‐posterior probabilities satisfy
\begin{IEEEeqnarray}{Cl}
p\bigl(b_i = 0 \mid \mathbf{y}\bigr)
= \sum_{\ell\in\mathcal{S}_i^0}\eta_{k,\ell},\IEEEyesnumber\IEEEyessubnumber*\\
p\bigl(b_i = 1 \mid \mathbf{y}\bigr)
= \sum_{\ell\in\mathcal{S}_i^1}\eta_{k,\ell}.
\end{IEEEeqnarray}
The LLR for the \(i\)-th bit of user \(k\) is defined as
\begin{equation}
	\label{eq:LLR_from_posteriors}
	\mathrm{LLR}_k(i)
	= \ln
	\frac{p(b_i=0\mid\mathbf{y})}
	{p(b_i=1\mid\mathbf{y})}
	= \ln
	\frac{\displaystyle\sum_{\ell\in\mathcal{S}_i^0}\eta_{k,\ell}}
	{\displaystyle\sum_{\ell\in\mathcal{S}_i^1}\eta_{k,\ell}}.
\end{equation}


The NCS-IGA algorithm is summarized in \autoref{algorithmn_NCSIGA}.
\begin{algorithm}[!hbt]
	\KwData{Channel matrix $\mathbf{H}$, received signal $\mathbf{y}$, the constellation set $\mathbf{c}$, damping factor $\alpha$, noise power $\sigma_z^2$, maximum iteration number $T$.}
	\KwResult{The LLR}
	\SetKwProg{Fn}{Function}{:}{}
	\Fn{NCS-IGA ($\mathbf{H},\mathbf{y},\mathbf{c},\alpha,\sigma_z^2,T$)}{
		$\boldsymbol{\lambda}_n \gets\mathbf{0}, \boldsymbol{\Lambda}_n \gets \mathbf{-1}$\;
		
		\For{$t \gets 1$ \textbf{to} $T$}{
			Calculate ${{\mu}}_n$ for all auxiliary manifolds as \eqref{eq: mu_q} \;
			Calculate the beliefs for all auxiliary manifolds as \eqref{eq: beliefs 2}\;
			Calculate the symbol probability as \eqref{eq:symbol probability}\;
			Calculate the beliefs for the extra auxiliary manifold as \eqref{eq:extra beliefs}\;
			Update the parameters as \eqref{eq:para_upd NonLinear damp}\;
		}
		Calculate the LLR as \eqref{eq:LLR_from_posteriors}\;
		\KwRet{$\mathrm{LLR}$}\;
	}
	\caption{NCS-IGA for Signal Detection}
	\label{algorithmn_NCSIGA}
\end{algorithm}

\subsection{Complexity Analysis}

The overall complexity of NCS‑IGA consists of two parts: the CS‑IGA backbone and the extra non‑linear marginalization steps.

\textbf{CS‑IGA backbone.}  As shown in Section \ref{CS-IGA Linear}, the core CS‑IGA iterations computing \(\mu_n\), \(\boldsymbol{\xi}_n\), \(\boldsymbol{\Xi}_n\) and updating \(\boldsymbol{\lambda}_n,\boldsymbol{\Lambda}_n\) incur 
\(
\mathcal{O}(TN^2)
\)
operations over \(T\) iterations and \(N\) auxiliary manifolds.

\textbf{Extra steps per iteration.}  For each user \(k=1,\dots,N\):
\begin{enumerate}
	\item \emph{posterior symbol probability computation} via \eqref{eq:symbol probability}:
	\[
	\{\eta_{k,1},\dots,\eta_{k,L}\}
	\;=\;
	p_{pr,k}(x^{(\ell)})\,p(x^{(\ell)};\hat{\boldsymbol{\theta}}_0,\hat{\boldsymbol{\Theta}}_0)
	\quad\forall \ell,
	\]
	which takes \(\mathcal{O}(L)\) evaluations of the Gaussian factor.
	\item \emph{\(m\)‑projection of extra auxiliary manifold} via \eqref{eq:mean and var of extra AM}:
	\[
	\tilde{\mu}_{0,k}
	= \sum_{\ell=1}^L \eta_{k,\ell}\,c_\ell,
	\quad
	\tilde{\Sigma}_{0,k}
	= \sum_{\ell=1}^L \eta_{k,\ell}\,\bigl|c_\ell-\tilde{\mu}_{0,k}\bigr|^2,
	\]
	which is also \(\mathcal{O}(L)\).
\end{enumerate}
Hence the extra cost per user  is \(\mathcal{O}(L)\), and for all \(N\) users \(\mathcal{O}(N\,L)\).  With \(T\) iterations, this becomes
\(
\mathcal{O}(T\,N\,L).
\)

Combining these, the total complexity of NCS-IGA is $\mathcal{O}(TN^2+TNL)$.

\noindent \textbf{Comparison with other algorithms:}
\begin{itemize}
	\item \emph{Bayes‑Optimal AMP} with discrete priors  has per‐iteration cost \(\mathcal{O}(MN + NL)\), where the \(MN\) term arises from linear steps and \(NL\) from extra steps in nonlinear scheme.  When \(M\gg N\), which is typical in XL‑MIMO, the complexity of the NCS‑IGA is significantly lower.
	\item \emph{Traditional IGA} \cite{IGSiDet} for discrete priors needs \(\mathcal{O}(MNL)\) operations per iteration, since each auxiliary update involves both matrix operations \(MN\) and summing over \(L\) symbols.  Again, the complexity of NCS‑IGA is lower when \(M\) is large.
	\item \emph{EP} has per iteration costs \(\mathcal{O}(N^3 + \,L\,N)\) \cite{EP}, dominated by the cubic matrix inversion.  Even the complexity of EP with a single iteration exceeds that of the LMMSE, whereas NCS‑IGA remains near \(\mathcal{O}(N^2)\) per iteration.
\end{itemize}

Thus, among current non‑linear detectors for XL‑MIMO, NCS‑IGA offers the lowest computational complexity.

%
%
%

\section{Simulation Results}\label{sim}
\subsection{Simulation Settings}

All simulations employ the QuaDriGa \cite{quadriga} channel generator under the ``3GPP-38.901-UMA–NLOS'' scenario. 
We consider a uniform planar array (UPA) at the BS with $M=N_h\times N_v = 32\times16=512$ antennas, severing $128$ single-antenna users. Users are distributed randomly within a $120^\circ$ azimuth sector of radius 200 meters around the BS, whose coordinates are $(0,0,25)$ meters. The large‐scale path loss and spatial correlation follow the 3GPP specified parameters for urban microcells with non‐line‐of‐sight.

We employ a center frequency of $f_c = 6.7$ GHz, $N_c = 2048$ subcarriers, subcarrier spacing $\Delta f = 30$ kHz, and a cyclic prefix length of $M_g = 144$ for the OFDM system. The channel matrix is normalized such that $\mathbb{E}[\|\mathbf{H}\|_F^2]=M$, and each transmitted symbol has unit average power.

We evaluate both linear and non‐linear detectors over coded data.  The channel code is an LDPC code from the 5G standard, with rate \(R=3/4\) for the linear CS‑IGA receiver and \(R=1/2\) for the non‐linear NCS‑IGA receiver.  We test QPSK, 16‑QAM, and 64‑QAM modulation.  User scheduling selects the best \(N=128\) users out of 512 candidates based on instantaneous channel correlation matrix $\mathbf{H}^H\mathbf{H}$.

Table \ref{tab:SimPara} summarizes the key system and channel parameters.

\begin{table}[htbp]
	\centering
	\caption{Parameter Setting of the QuaDriGa Channel Model}
	\label{tab:SimPara}
	\renewcommand\arraystretch{0.75}
	\begin{tabular}{lc}
		\hline
		Parameter & Value\\
		\hline
		Number of BS antennas \(M = N_h \times N_v\) 
		& \(512 = 32 \times 16\) \\
		Number of user antennas \(N\) 
		& 128 \\
		Modulation schemes 
		& QPSK, 16‑QAM, 64‑QAM \\
		Coding scheme 
		& LDPC (5G standard) \\
		Code rates 
		& \(R = 3/4\) , \(1/2\) \\
		Center frequency \(f_c\) 
		& 6.7 GHz \\
		Number of subcarriers \(N_c\) 
		& 2048 \\
		Subcarrier spacing \(\Delta f\) 
		& 30 kHz \\
		Cyclic prefix length \(M_g\) 
		& 144 \\
		Scheduling pool size 
		& 512 users \\
		\hline
	\end{tabular}
\end{table}

\subsection{Performance of the CS-IGA}
We first present the convergence and BER performance of CS‑IGA compared to the Bayes‑optimal AMP and the conventional IGA in the linear detection scheme. The modulation formats are 16‑QAM and 64‑QAM. All methods employ soft demodulation followed by LDPC decoding. Since user scheduling is performed before detection, the channel correlation matrix \(\mathbf{H}^H\mathbf{H}\) is well conditioned; these iterative algorithms therefore require only a few iterations to converge.

For the 16‑QAM case, \autoref{fig:iciter10db16qam} shows the convergence behavior of CS‑IGA at a signal‑to‑noise ratio (SNR) of 13~dB.From this figure, we can observe that CS‑IGA converges in 4 iterations to achieve the lowest BER, while AMP and IGA require 6 and 7 iterations, respectively. At BER \(=10^{-3}\), CS‑IGA reduces the required number of iterations by one compared to IGA. \autoref{fig:ic16qamiter4} plots the BER performance of the three detectors when each is limited to 4 iterations. CS‑IGA achieves performance essentially identical to LMMSE detection within these 4 iterations. At BER \(=10^{-5}\), CS‑IGA achieves SNR gains of approximately 1.5 dB over AMP and 2.2 dB over IGA.

%
\begin{figure}[h]
	\centering
	\includegraphics[width=0.93\linewidth]{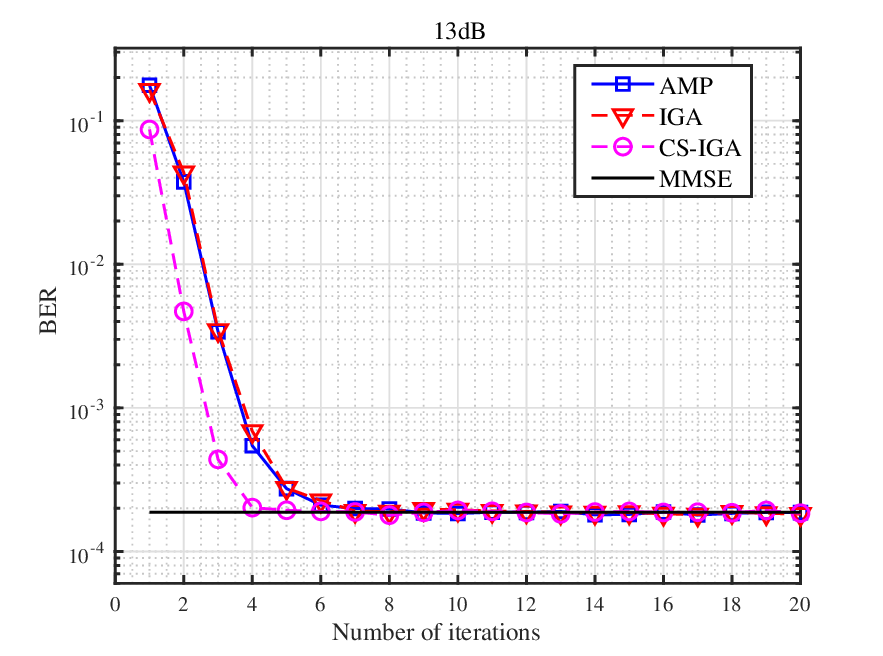}
	\caption{BER versus iteration number for different approaches at SNR=13dB  under 16QAM modulation.}
	\label{fig:iciter10db16qam}
\end{figure}
\begin{figure}[h]
	\centering
	\includegraphics[width=0.93\linewidth]{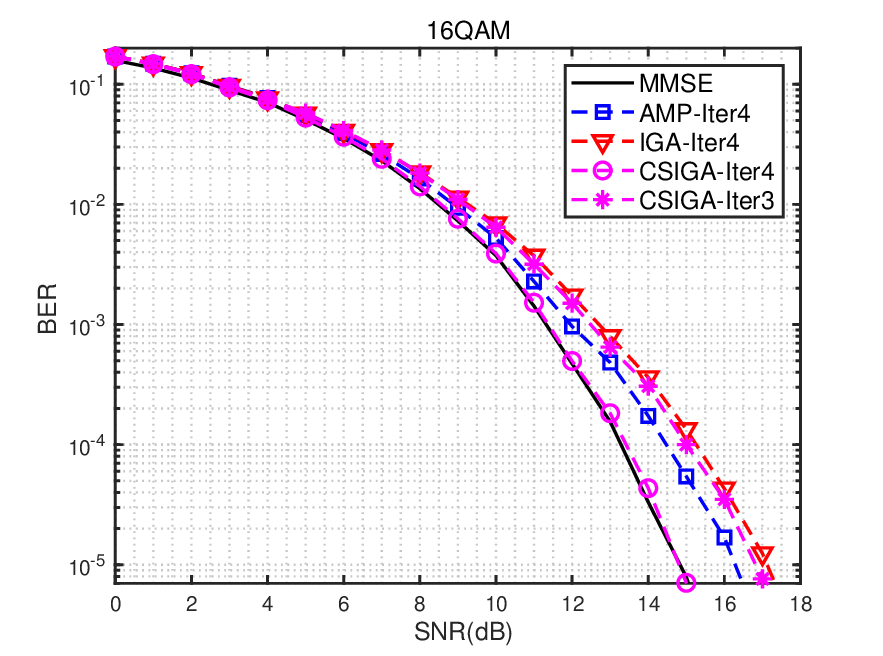}
	\caption{BER versus SNR for different approaches at 3 and 4 iterations under 16QAM modulation.}
	\label{fig:ic16qamiter4}
\end{figure}

Under 64‑QAM modulation, \autoref{fig:iciter14db64qam} and \autoref{fig:ic64qamiter5} show the convergence and BER performance of the three detectors. From \autoref{fig:iciter14db64qam}, CS‑IGA converges faster than AMP and IGA. 
At BER $=10^{-3}$, it requires one fewer iteration than either AMP or IGA. In \autoref{fig:ic64qamiter5}, with the same iteration count, CS‑IGA achieves the best BER performance and most closely matches the MMSE bound. Furthermore, the BER of CS‑IGA at 5 iterations outperforms that of AMP and IGA at 6 iterations. 
At BER $=10^{-5}$, CS-IGA gains approximately 1.2 dB over AMP and 2 dB over IGA. We also observe that, compared to 16‑QAM, all iterative methods require more iterations to converge under 64‑QAM.
These linear‑scheme results were first reported in our conference paper.

\begin{figure}[h]
	\centering
	\includegraphics[width=0.93\linewidth]{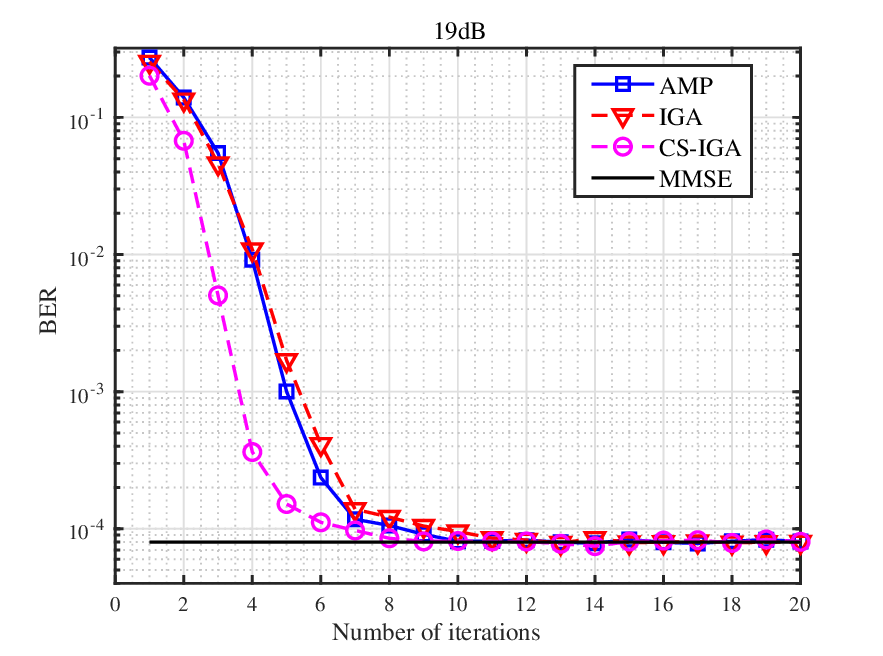}
	\caption{BER versus iteration number for different approaches at SNR=19dB  under 64QAM modulation.}
	\label{fig:iciter14db64qam}
\end{figure}
\begin{figure}[h]
	\centering
	\includegraphics[width=0.93\linewidth]{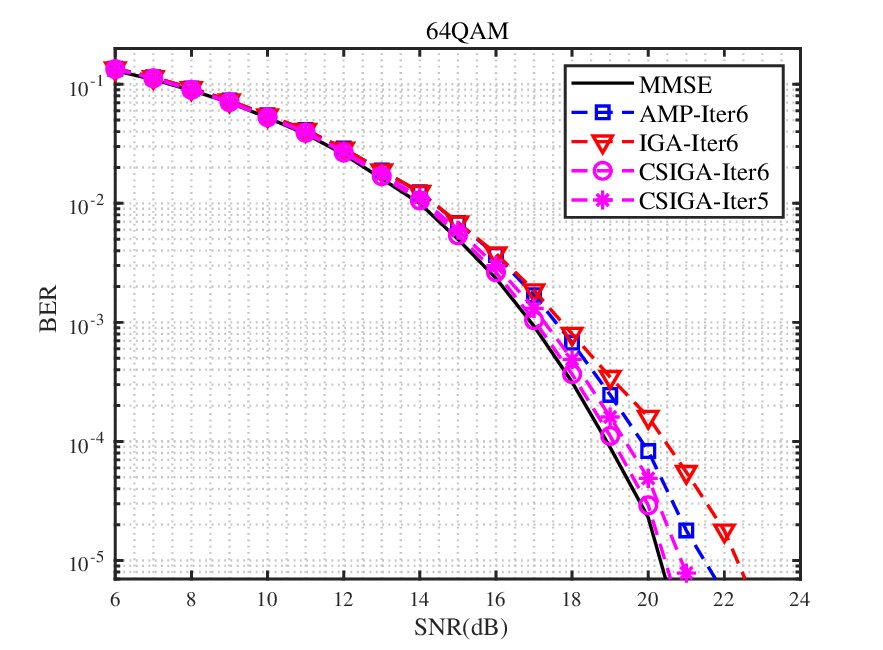}
	\caption{BER versus SNR for different approaches at 5 and 6 iterations under 64QAM modulation.}
	\label{fig:ic64qamiter5}
\end{figure}

\subsection{Performance of the NCS-IGA}
Next, we present the convergence and BER performance of NCS‑IGA compared to the Bayes‑optimal AMP, the conventional IGA, and EP in the non‑linear detection scheme. The modulation formats are QPSK, 16‑QAM, and 64‑QAM. All methods compute LLRs and then perform LDPC decoding.
\autoref{fig:niciter5db_QPSK} shows the convergence behavior of NCS-IGA at SNR = 5 dB under QPSK.
From this figure, we observe that both NCS‑IGA and EP converge in 2 iterations and achieve the lowest BER, while AMP and IGA require about 3 iterations. 
At BER = $10^{-5}$, NCS‑IGA saves one iteration compared to AMP and IGA. Notably, all of these iterative non‑linear detectors outperform the LMMSE bound. 
\autoref{fig:nic_QPSK_iter2} shows the BER versus SNR at 1 and 2 iterations: at the same iteration counts, NCS‑IGA and EP deliver the best BER performance, and at BER = $10^{-5}$ NCS‑IGA gains approximately 0.7 dB over AMP and IGA.
\begin{figure}[h]
	\centering
	\includegraphics[width=0.93\linewidth]{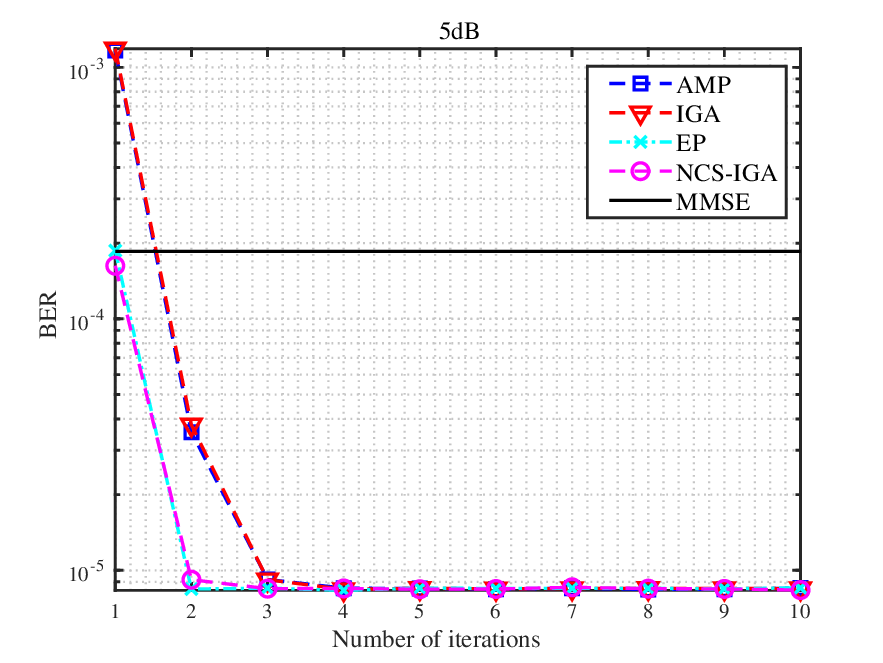}
	\caption{BER versus iteration number for different approaches at SNR=5dB  under QPSK modulation.}
	\label{fig:niciter5db_QPSK}
\end{figure}
\begin{figure}[h]
	\centering
	\includegraphics[width=0.93\linewidth]{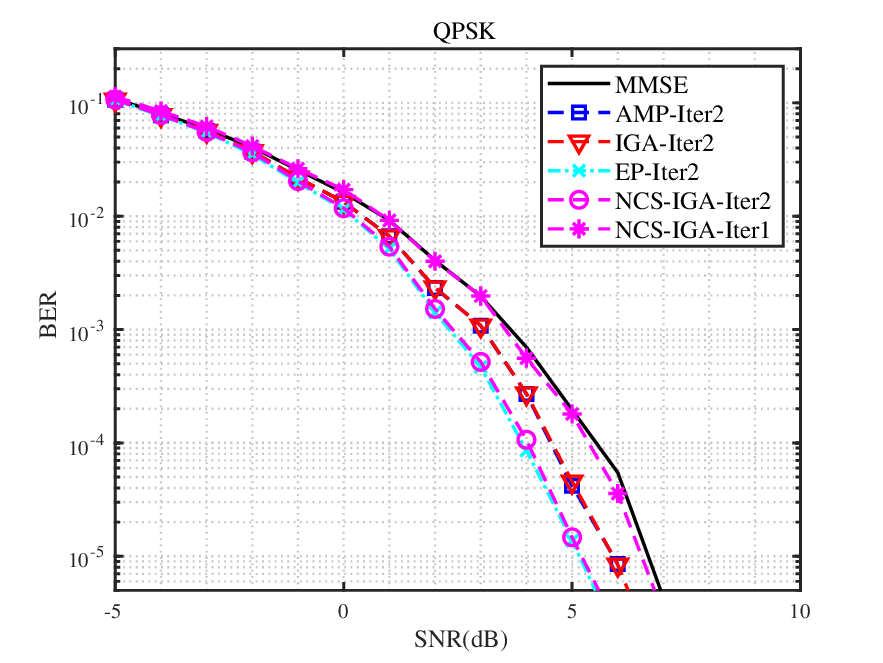}
	\caption{BER versus SNR for different approaches at 2 and 1 iterations under QPSK modulation.}
	\label{fig:nic_QPSK_iter2}
\end{figure}
The convergence and BER performance of NCS-IGA under 16-QAM modulation are shown in \autoref{fig:niciter11db16qam} and \autoref{fig:nic16qamiter4}.
In \autoref{fig:niciter11db16qam}, EP converges in 2 iterations, while NCS‑IGA, IGA, and Bayes‑optimal AMP each require about 5 iterations. For iteration counts below 5, NCS‑IGA achieves a clear BER advantage over both AMP and IGA. 
\autoref{fig:nic16qamiter4} shows that, at the same iteration count, NCS‑IGA have comparable BER performance with AMP and IGA.

\begin{figure}[h]
	\centering
	\includegraphics[width=0.93\linewidth]{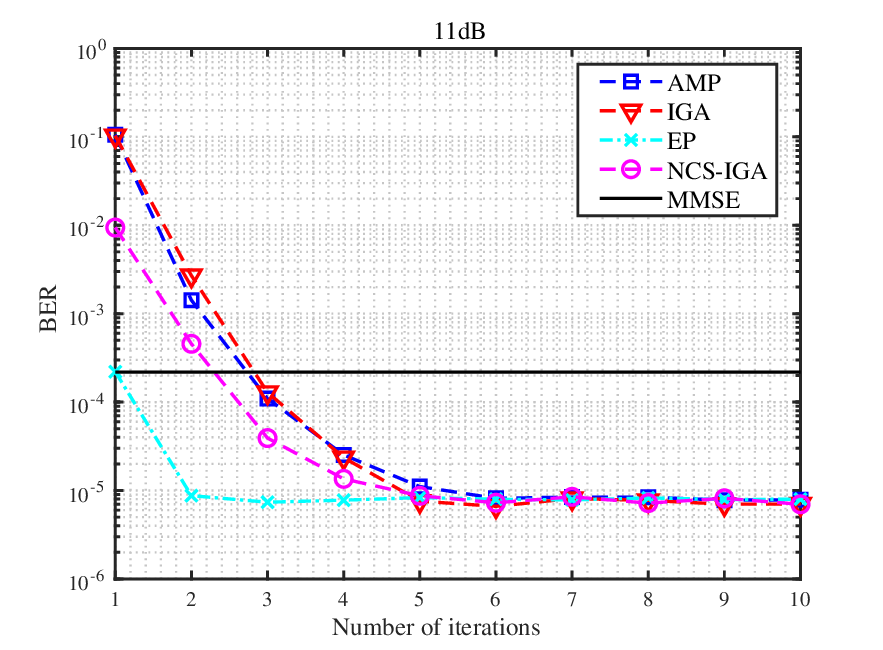}
	\caption{BER versus iteration number for different approaches at SNR=11dB  under 16-QAM modulation.}
	\label{fig:niciter11db16qam}
\end{figure}
\begin{figure}[h]
	\centering
	\includegraphics[width=0.93\linewidth]{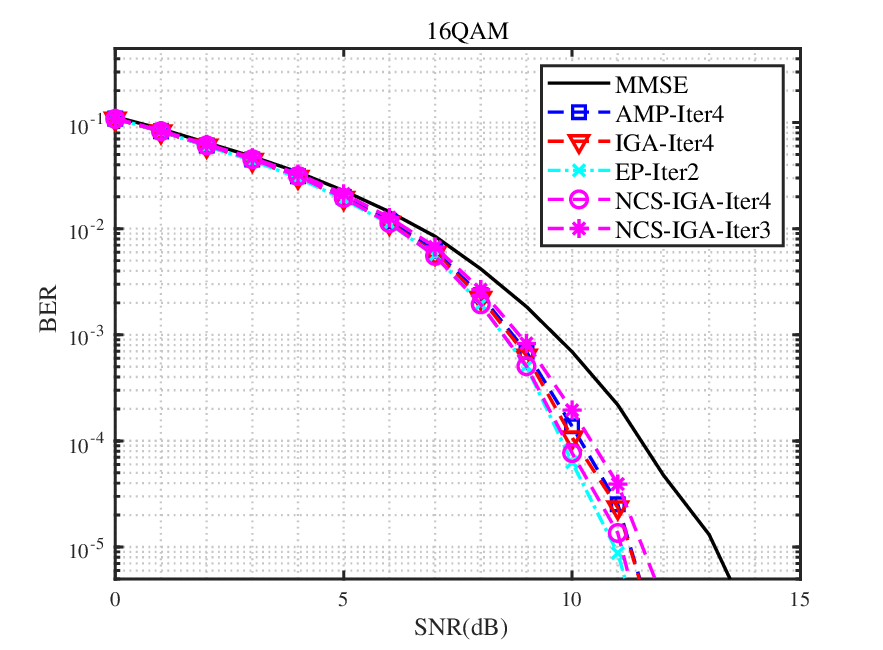}
	\caption{BER versus SNR for different approaches at 4 and 3 iterations under 16-QAM modulation.}
	\label{fig:nic16qamiter4}
\end{figure}

Under 64-QAM modulation (\autoref{fig:niciter16db64qam} and \autoref{fig:nic64qamiter6}), EP converges in approximately 4 iterations, whereas NCS-IGA, IGA, and AMP each need around 9 iterations.
In \autoref{fig:nic64qamiter6}, at low SNR all four methods yield nearly identical BER; at high SNR, AMP converges slightly faster than NCS-IGA and IGA.
At BER $=10^{-5}$ and 6 iterations, NCS-IGA incurs about a 0.5 dB loss relative to AMP.
\begin{figure}[h]
	\centering
	\includegraphics[width=0.93\linewidth]{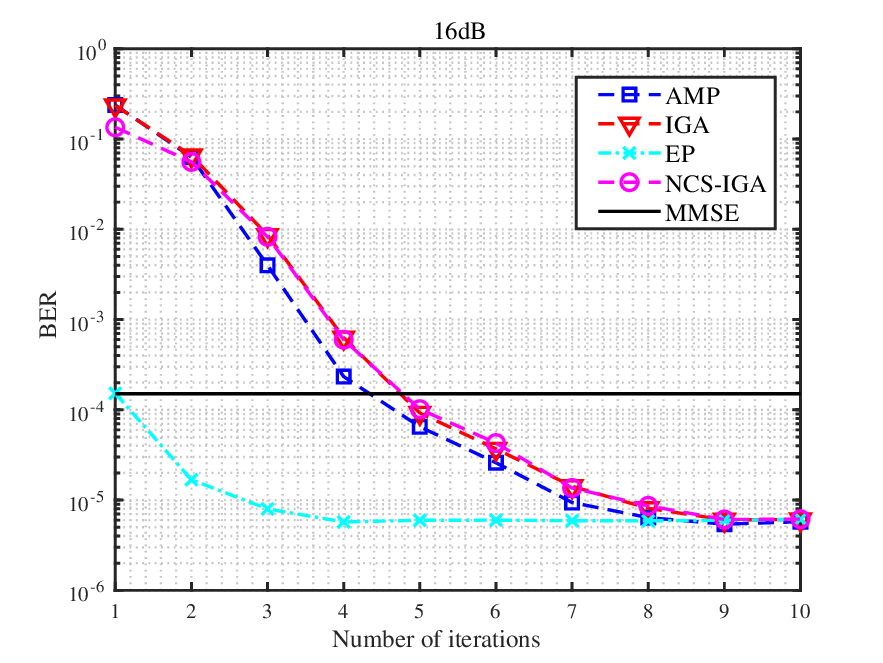}
	\caption{BER versus iteration number for different approaches at SNR=16dB  under 64-QAM modulation.}
	\label{fig:niciter16db64qam}
\end{figure}
\begin{figure}[h]
	\centering
	\includegraphics[width=0.93\linewidth]{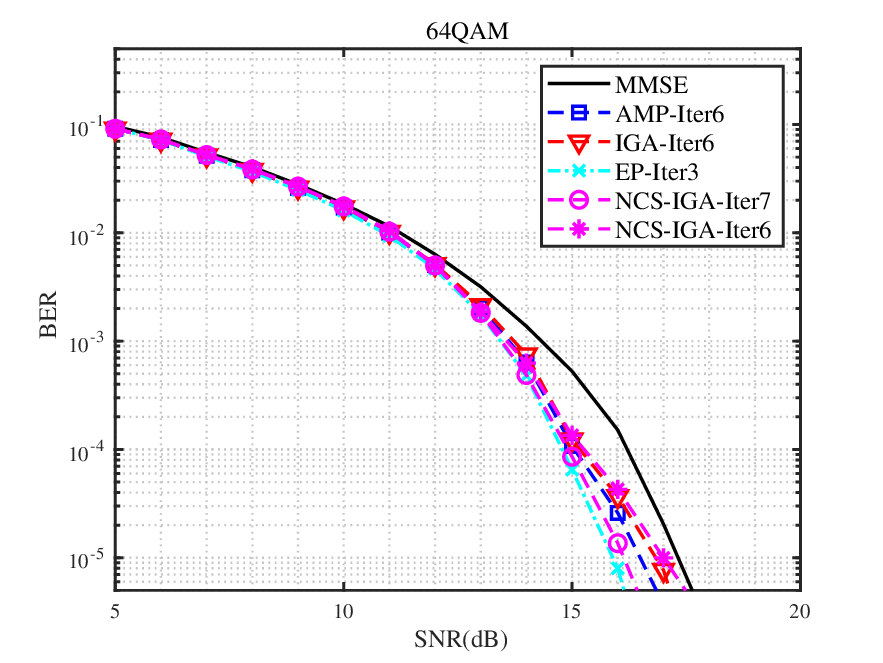}
	\caption{BER versus SNR for different approaches at 7 and 6 iterations under 64-QAM modulation.}
	\label{fig:nic64qamiter6}
\end{figure}

\section{Conclusion}\label{conclusion}
In this paper, we have addressed the critical challenge of high-complexity signal detection in uplink XL-MIMO systems. We introduced a novel iterative detector, the CS-IGA, designed to achieve both low computational cost and rapid convergence.
We further extended our method to handle discrete constellations in the nonlinear CS-IGA variant NCS-IGA. By integrating the symbol-wise moment matching directly into the geometric framework, NCS-IGA efficiently approximates the marginal \emph{a posteriori} probabilities without requiring external processing loops. Our simulation results, conducted using realistic 3GPP channel models, have validated the theoretical benefits of our approach. Both CS-IGA and NCS-IGA consistently demonstrated faster convergence and matched or exceeded the BER performance of state-of-the-art AMP and IGA detectors, while operating at a fraction of the computational cost.

\appendices
\section{Proof of Theorem 1}\label{Appendix A}
From the relationship between $\boldsymbol{\Theta}_n$ and $\boldsymbol{\Lambda}_n$ in \eqref{eq: natural_para_AM} , and the transform from natural parameters to expectation parameters in \eqref{eq:legendre1} and \eqref{eq:legendre2}, we have the following equation
\begin{IEEEeqnarray}{Cl}
	\mathbf{P}_{1n}\boldsymbol{\Sigma}_n\mathbf{P}_{1n}^H&=-\mathbf{P}_{1n}\boldsymbol{\Theta}_n^{-1}\mathbf{P}_{1n}^H,\notag\\
	&=(\mathbf{P}_{1n}(\boldsymbol{\Lambda}_n+\mathbf{C}_n+\mathbf{D})\mathbf{P}_{1n}^H)^{-1},\notag\\
	&= \left(\begin{array}{cc}
		d_n+\Lambda_n & \bar{\mathbf{k}}_n^H\\
		\bar{\mathbf{k}}_n&\bar{\mathbf{D}}_{ n}+\bar{\boldsymbol{\Lambda}}_{n}
	\end{array}\right)^{-1},
\end{IEEEeqnarray}
Using the block matrix inversion lemma, the inversion is calculated as
\begin{IEEEeqnarray}{Cl}
	&\mathbf{P}_{1n}\boldsymbol{\Sigma}_n\mathbf{P}_{1n}^H =  \left[\begin{array}{cc}
		r_n & \mathbf{m}_n^H \\ 
		\mathbf{m}_n  & \mathbf{J}_n
	\end{array}\right],
\end{IEEEeqnarray}
where $r_n$ and $\mathbf{m}_n$ are denoted by
\begin{IEEEeqnarray}{Cl}
	& r_n = (\Lambda_n+d_n-\bar{\mathbf{k}}_n^H\check{\boldsymbol{\Lambda}}_n\bar{\mathbf{k}}_n)^{-1},
	\IEEEyessubnumber\\
	& \mathbf{m}_n = -{r_n}\check{\boldsymbol{\Lambda}}_n\bar{\mathbf{k}}_n,
	\IEEEyessubnumber
\end{IEEEeqnarray}
and $\mathbf{J}_n$ is expressed as
\begin{IEEEeqnarray}{Cl}
	\mathbf{J}_n = \left(\check{\boldsymbol{\Lambda}}_n^{-1}-\frac{1}{(\Lambda_n+d_n)}\bar{\mathbf{k}}_n\bar{\mathbf{k}}_n^H\right)^{-1}.
\end{IEEEeqnarray}
Since $\bar{\mathbf{k}}_n\bar{\mathbf{k}}_n^H$ is a rank-1 matrix, the inversion can be easily obtained by Sherman-Morrison formula as
\begin{IEEEeqnarray}{Cl}
	\mathbf{J}_n &= \check{\boldsymbol{\Lambda}}_n+\frac{\frac{1}{(\Lambda_n+d_n)}\check{\boldsymbol{\Lambda}}_n\bar{\mathbf{k}}_n\bar{\mathbf k}_n^H\check{\boldsymbol{\Lambda}}_n}{1-\frac{1}{(\Lambda_n+d_n)}\bar{\mathbf{k}}_n^H\check{\boldsymbol{\Lambda}}_n\bar{\mathbf k}_n},\notag\\
	&=\check{\boldsymbol{\Lambda}}_n+{r_n}\check{\boldsymbol{\Lambda}}_n\bar{\mathbf{k}}_n\bar{\mathbf{k}}_n^H\check{\boldsymbol{\Lambda}}_n.
\end{IEEEeqnarray} 
Then, $\mathbf{P}_{1n}\boldsymbol{\Sigma}_n\mathbf{P}_{1n}^H$ can be written as
\begin{IEEEeqnarray}{Cl}
	\mathbf{P}_{1n}\boldsymbol{\Sigma}_n\mathbf{P}_{1n}^H =  \left[\begin{array}{cc}
		r_n & \mathbf{m}_n^H \\ 
		\mathbf{m}_n  & \check{\boldsymbol{\Lambda}}_n+{r_n}\check{\boldsymbol{\Lambda}}_n\bar{\mathbf{k}}_n\bar{\mathbf{k}}_n^H\check{\boldsymbol{\Lambda}}_n
	\end{array}\right].
\end{IEEEeqnarray}

Then we calculate $\mathbf{P}_{1n}\boldsymbol{\mu}_n$ based on $\mathbf{P}_{1n}\boldsymbol{\Sigma}_n\mathbf{P}_{1n}^H$. Based on \eqref{eq:legendre2}, $\boldsymbol{\mu}_n$ is easily obtained by
\begin{IEEEeqnarray}{Cl}
	\mathbf{P}_{1n}\boldsymbol{\mu}_n &=\mathbf{P}_{1n}\boldsymbol{\Sigma}_n(\boldsymbol{\lambda}_n+\mathbf{b}_n),\notag\\
	&= (\boldsymbol{\Lambda}_n+\mathbf{D}+\mathbf{C}_n)^{-1}(\boldsymbol{\lambda}_n+\mathbf{b}_n),\notag\\
	&= \begin{bmatrix}
		r_n&\mathbf{m}_n^H\\
		\mathbf{m}_n&\mathbf{J}_n
	\end{bmatrix}\begin{bmatrix}
	\lambda_n+\frac{1}{\sigma_z^2}\mathbf{h}_n^H\mathbf{y}\\
	\bar{\boldsymbol{\lambda}}_n
	\end{bmatrix}.
\end{IEEEeqnarray}
The $n-$th element of $\boldsymbol{\mu}_n$ is calculated as
\begin{IEEEeqnarray}{Cl}
	\mu_n &= r_n(\lambda_{ n}+r_n\frac{1}{\sigma_z^2}\mathbf{h}_n^H\mathbf{y})+\mathbf{m}_n^H\bar{\boldsymbol{\lambda}}_n,\notag\\
	&= \frac{r_n}{\sigma_z^2}\mathbf{h}_n^H\mathbf{y}+r_n\lambda_n-r_n\bar{\mathbf{k}}_n^H\check{\boldsymbol{\Lambda}}_n^H\bar{\boldsymbol{\lambda}}_n.
\end{IEEEeqnarray}
Notice that $\check{\boldsymbol{\Lambda}}_n$ is a real diagonal matrix, so $\check{\boldsymbol{\Lambda}}_n^H = \check{\boldsymbol{\Lambda}}_n$. For simplicity, we define $v_n = \bar{\mathbf{k}}_n^H\check{\boldsymbol{\Lambda}}_n\bar{\boldsymbol{\lambda}}_n$, then $\mu_n$ is rewritten as
\begin{IEEEeqnarray}{Cl}
	\mu_n = r_n\frac{1}{\sigma_z^2}\mathbf{h}_n^H\mathbf{y}+r_n\lambda_n-r_nv_n.
\end{IEEEeqnarray}
The rest elements $\bar{\boldsymbol{\mu}}_n$ are calculated as
\begin{IEEEeqnarray}{Cl}
	\hspace{-10pt}
	\bar{\boldsymbol{\mu}}_n &= \mathbf{m}_n(\lambda_n+\frac{1}{\sigma_z^2}\mathbf{h}_n^H\mathbf{y})+\mathbf{J}_n\bar{\boldsymbol{\lambda}}_n,\notag\\
	&= -r_n\frac{1}{\sigma_z^2}\check{\boldsymbol{\Lambda}}_n\bar{\mathbf{k}}_n\mathbf{h}_n^H\mathbf{y}-r_n\check{\boldsymbol{\Lambda}}_n\bar{\mathbf{k}}_n\lambda_n\notag\\
	&~~+ \check{\boldsymbol{\Lambda}}_n\bar{\boldsymbol{\lambda}}_n+r_n\check{\boldsymbol{\Lambda}}_n\bar{\mathbf{k}}_n\bar{\mathbf{k}}_n^H\check{\boldsymbol{\Lambda}}_n\bar{\boldsymbol{\lambda}}_n,\notag\\
	&=-\check{\boldsymbol{\Lambda}}_n\bar{\mathbf{k}}_n(\frac{r_n}{\sigma_z^2}\mathbf{h}_n^H\mathbf{y}+r_n\lambda_n-r_n\bar{\mathbf{k}}_n^H\check{\boldsymbol{\Lambda}}_n\bar{\boldsymbol{\lambda}}_n)+\check{\boldsymbol{\Lambda}}_n\bar{\boldsymbol{\lambda}}_n,\notag\\
	&= -\mu_n\check{\boldsymbol{\Lambda}}_n\bar{\mathbf{k}}_n+\check{\boldsymbol{\Lambda}}_n\bar{\boldsymbol{\lambda}}_n.
\end{IEEEeqnarray}
Then $\mathbf{P}_{1n}\boldsymbol{\mu}_n$ is denoted as
\begin{IEEEeqnarray}{Cl}
	\mathbf{P}_{1n}\boldsymbol{\mu}_n=\begin{bmatrix}
		\mu_n\\
		\bar{\boldsymbol{\mu}}_n
	\end{bmatrix}=\left[\begin{array}{c}
	\frac{r_n}{\sigma_z^2}\mathbf{h}_n^H\mathbf{y}+r_n\lambda_n-r_nv_n\\
	-\mu_n\check{\boldsymbol{\Lambda}}_n\bar{\mathbf{k}}_n+\check{\boldsymbol{\Lambda}}_n\bar{\boldsymbol{\lambda}}_{n}
	\end{array}\right].
\end{IEEEeqnarray}
This complete the proof.

\section{Proof of Theorem 2}\label{Appendix B}
Upon convergence, the algorithm satisfies both the \(e\)- and \(m\)-conditions, the relationship of expectation parameters at the fixed point is given by \eqref{eq: fixed point mean} and \eqref{eq: fixed point var}, which leads to
\begin{IEEEeqnarray}{Cl}
	\label{eq: fixed lambda}
	\sum_{n=1}^N\boldsymbol{\lambda}_n^{\bullet} &= (N-1)\boldsymbol{\lambda}_0^{\bullet},
	\IEEEyesnumber\IEEEyessubnumber\\
	\sum_{n=1}^{N}\boldsymbol{\Lambda}_n^{\bullet}&=(N-1)\boldsymbol{\Lambda}_0^{\bullet}.
	\IEEEyessubnumber
\end{IEEEeqnarray}
Considering $\boldsymbol{\mu}_n$ satisfies
\begin{IEEEeqnarray}{Cl}
	\boldsymbol{\mu}_n &= -\boldsymbol{\Theta}_n^{-1}\boldsymbol{\theta}_n\notag\\ &=(\boldsymbol{\Lambda}_n+\mathbf{C}_n+\mathbf{D})^{-1}(\boldsymbol{\lambda}_n+\mathbf{b}_n).
\end{IEEEeqnarray}
By taking sum over natural parameters of all AMs, we have the following property 
\begin{IEEEeqnarray}{Cl}
	\label{eq: sum lambda}
	\sum_{n=1}^{N}\boldsymbol{\lambda}_n^{\bullet}&=\sum_{n=1}^{N}(\boldsymbol{\Lambda}_n^{\bullet}+\mathbf{C}_n+\mathbf{D})\boldsymbol{\mu}^{\bullet}-\sum_{n=1}^{N}\mathbf{b}_n,\notag\\
	&= \left(\sum_{n=1}^{N}\boldsymbol{\Lambda}_n^{\bullet}+\mathbf{K}-\mathbf{I}\odot\mathbf{K}+\mathbf{D}\right)\boldsymbol{\mu}^{\bullet}-\mathbf{H}^H\mathbf{y},\notag\\
	&=\left(\sum_{n=1}^{N}\boldsymbol{\Lambda}_n^{\bullet}+\mathbf{K}+\mathbf{I}\right)\boldsymbol{\mu}^{\bullet}-\mathbf{H}^H\mathbf{y}
\end{IEEEeqnarray}
where $\mathbf{K}=\sigma_z^{-2}\mathbf{H}^H\mathbf{H}$. Substituting \eqref{eq: fixed lambda} into \eqref{eq: sum lambda}, we have
\begin{IEEEeqnarray}{Cl}
	\label{eq: N-1 lamdba}
	(N-1)\boldsymbol{\lambda}_0^{\bullet}=(N-1)\boldsymbol{\Lambda}_0^{\bullet}\boldsymbol{\mu}^{\bullet}+(\mathbf{K}+\mathbf I)\boldsymbol{\mu}^{\bullet}-\mathbf{H}^H\mathbf{y}.
\end{IEEEeqnarray}
Since $\boldsymbol{\mu}^{\bullet}=(\boldsymbol{\Lambda}_0^{\bullet})^{-1}\boldsymbol{\lambda}^{\bullet}$, \eqref{eq: N-1 lamdba} can be rewritten as
\begin{IEEEeqnarray}{Cl}
	\mathbf{H}^H\mathbf{y} = (\mathbf{K}+\mathbf I)\boldsymbol{\mu}^{\bullet}.
\end{IEEEeqnarray}
Then $\boldsymbol{\mu}^{\bullet}$ is obtained by 
\begin{IEEEeqnarray}{Cl}
	\boldsymbol{\mu}^{\bullet} &= (\mathbf{K} +\mathbf{I})^{-1}\mathbf{H}^H\mathbf y,
	\notag\\
	&=(\sigma_z^{-2}\mathbf{H}^H\mathbf{H}+\mathbf{I})^{-1}\mathbf{H}^H\mathbf{y},
\end{IEEEeqnarray}
which is equivalent to the \textit{a posteriori} mean in \eqref{eq:posterior mean} or the LMMSE mean. This complete the proof.
\small
\bibliographystyle{IEEEtran}
\bibliography{IEEEabrv,reference}

%

\end{document}